\documentclass[twoside,11pt]{article}

\usepackage{epsfig}
\usepackage{amssymb}
\usepackage{natbib}
\usepackage{hyperref}
\usepackage{graphicx}
\usepackage{thumbpdf}
\usepackage{amsfonts,amstext,amsmath,amsthm}
\usepackage{mathtools,bbm}
\usepackage{bm}
\usepackage{enumitem}
\usepackage{subcaption}
\usepackage{booktabs}
\usepackage{algorithm2e}
\RestyleAlgo{ruled}
\usepackage{mathabx}
\DeclarePairedDelimiter\norm\lVert\rVert
\usepackage{accents}
\usepackage{color}
\usepackage{rotating}
\usepackage[utf8]{inputenc}
\usepackage{tikz}
\usepackage{multirow}
\usepackage{float}
\usepackage[labelfont=bf, width=15cm, font=small]{caption}
\usepackage[onehalfspacing]{setspace}
\usepackage[left=3cm, right=3cm, top=2.5cm, bottom=3cm]{geometry}
\setcitestyle{authoryear,open={(},close={)}}
\newcommand*{\doi}[1]{\href{http://dx.doi.org/#1}{doi: #1}}
\usepackage{authblk}
\usepackage{multibib}
\newcites{APP}{References}

\DeclareMathOperator*{\argmin}{arg\,min}

\newcommand{\bbR}{\mathbb{R}}
\newcommand{\bbN}{\mathbb{N}}

\newcommand{\cH}{\mathcal{H}}
\newcommand{\cP}{\mathcal{P}}

\newcommand{\iid}{\stackrel{\mathrm{iid}}{\sim}}

\newcommand{\dto}{\stackrel{d}{\to}}

\newcommand{\bbP}{\mathbb{P}}
\newcommand{\bbE}{\mathbb{E}}

\newcommand{\cov}{\text{Cov}}

\newcommand{\var}{\text{Var}}
\newcommand{\cN}{\mathcal{N}}

\theoremstyle{definition}
\newtheorem{definition}{Definition}[section]
\newtheorem{proposition}{Proposition}[section]
\newtheorem{theorem}{Theorem}[section]

\newtheorem{lemma}{Lemma}[section]

\newtheorem{assumption}{Assumption}[section]
\newtheorem*{notation}{Notation}

\theoremstyle{plain}

\theoremstyle{remark}
\newtheorem{remark}{Remark}[section]

\newcommand{\bX}{\mathbf{X}}
\newcommand{\bY}{\mathbf{Y}}
\newcommand{\bZ}{\mathbf{Z}}

\newcommand{\cD}{\mathcal{D}}

\newcommand{\cL}{\mathcal{L}}

\newcommand{\cB}{\mathcal{B}}

\begin{document}

\title{Sparse Causal Effect Estimation using Two-Sample Summary Statistics in the Presence of Unmeasured Confounding}

\author[1]{Shimeng Huang}
\author[1]{Niklas Pfister}
% \author[2]{Jack Bowden}
\author[2,3,4]{Jack Bowden}
\affil[1]{Department of Mathematical Sciences, University of Copenhagen}
\affil[2]{Novo Nordisk Research Centre Oxford (NNRCO)}
\affil[3]{Biostatistics, Novo Nordisk}
\affil[4]{Exeter Medical School, University of Exeter}
% \affil[2]{Exeter Medical School, University of Exeter}

\date{October 18, 2024}

\maketitle

\begin{abstract}%
  Observational genome-wide association studies are now widely used for causal
inference in genetic epidemiology. To maintain privacy, such data is often only
publicly available as summary statistics, and often studies for the endogenous covariates
and the outcome are available separately. This has necessitated methods tailored
to two-sample summary statistics. Current state-of-the-art methods modify linear
instrumental variable (IV) regression---with genetic variants as
instruments---to account for unmeasured confounding. However, since the
endogenous covariates can be high dimensional, standard IV assumptions are generally
insufficient to identify all causal effects simultaneously. We ensure
identifiability by assuming the causal effects are sparse and propose a sparse
causal effect two-sample IV estimator, \texttt{spaceTSIV}, adapting the
\texttt{spaceIV} estimator by \citet{pfister2022identifiability} for two-sample
summary statistics. We provide two methods, based on L0- and L1-penalization,
respectively. We prove identifiability of the sparse causal effects in the
two-sample setting and consistency of \texttt{spaceTSIV}. The performance of
\texttt{spaceTSIV} is compared with existing two-sample IV methods in
simulations. Finally, we showcase our methods using real proteomic and
gene-expression data for drug-target discovery.

\end{abstract}

%%%%%%%%%%%%%%%%%%%%%%%%%%%%%%%%%%%%%%%%%%%%%%%%%%%%%%%%%%%%%%%%%%%%%%%%%%%%%%%%
\section{Introduction} \label{sec:intro}
%-------------------------------------------------------------------------------

The use of observational data to study the causal effects of covariate 
interventions on an outcome has seen a surge in popularity in many scientific areas. 
A primary example is genetic epidemiology, where a common research topic is to 
study the causal effects of genetically predictive phenotypic traits, 
such as a person's body mass index or their low density lipoprotein cholesterol, 
on downstream disease outcomes. 
This is often based on Mendelian randomization (MR)---that is, instrumental 
variable estimation (IV) with genetic variants 
being the instruments---to account for unmeasured confounding between the 
endogenous covariates and the outcome. However, due to privacy concerns, access to 
individual-level genetic data is highly regulated. 
To both preserve privacy and enable data sharing, public data repositories of 
genetic summary statistics are made available by various international 
genome-wide association study (GWAS) consortia. 
These summary statistics usually contain marginal effect estimates 
of single nucleotide polymorphisms (SNPs) on the phenotypic traits and 
disease outcomes, along with their standard errors, which are often themselves 
obtained from two separate GWAS. This is referred to as the ``two-sample summary 
statistics" setting. \citet{zhao2019two} discuss sufficient 
assumptions that enable consistent estimation under two-sample IV, 
specifically the homogeneity of the two samples. 
When the number of endogenous covariates under investigation is high dimensional 
or the instruments are highly correlated, a case in point being 
human gene expression phenotypes and genetic variants, there may be an 
insufficient number of strong and valid instruments to ensure the 
identifiability of the multivariable causal effects.

Lack of identifiability leads to poor estimation, or weak instrument bias. 
In the univariable two-sample summary statistics setting, 
\citet{bowden2019improving} develop heuristic weak-instrument robust 
inference strategies based on heterogeneity statistic 
estimating equations. Under the same setting, \citet{wang2022weak} further 
clarify the connection between these approaches and summary statistics 
analogues of the Anderson-Rubin (AR) test statistic \citep{anderson1949estimation} and Limited Information 
Maximum Likelihood (LIML).  
\citet{jingshu2020} further extend weak-instrument robust models to the 
multivariable case. Another way to circumvent the weak instrument 
problem is to employ principal component analysis (PCA).
Building on the work of \citet{batool2022disentangling}, \citet{patel2024robust} 
show how many individually weak variants could be fashioned into PCA scores 
with improved instrument strength.

An alternative strategy to tackle the lack of identifiability is to 
introduce sparsity assumption on the causal effects. 
This is often a reasonable assumption in MR studies, as it is usually the 
case that many endogenous traits do not have direct causal effects on the outcome. 
Under the assumptions of independent 
instruments and the number of instruments is no less than the number of 
covariates,
\citet{grant2022efficient} consider the use of L1 penalization on the causal effects
in multivariable MR models where one covariate is of special 
interest but the others are allowed to be penalized. In related works, \citet{rees2019robust}, 
\citet{zhao2020statistical} and \citet{grant2021pleiotropy} consider 
L1 penalization for individual instruments suspected to be invalid due 
to exclusion restriction violation, rather than penalization on the number 
of causal effects. 

In the one sample individual-level data setting, 
the identifiability conditions for sparse causal effects have been 
studied by \citet{pfister2022identifiability}, and they propose
a sparse causal effect estimator, \texttt{spaceIV}.
\citet{tang2023synthetic} also consider sparse causal effect 
identification and estimation under assumptions on the sparsity level and 
propose a synthetic two-stage regularized regression approach. 

We propose \texttt{spaceTSIV}, adapting the \texttt{spaceIV} 
estimator for two-sample summary statistics. We allow the IVs to be correlated 
by extending the adjustment method in \citet{wang2022weak}. 
Two specific approaches based on L0- and L1-penalization, respectively, 
are provided. We prove identifiability of the sparse causal effects 
and consistency of \texttt{spaceTSIV} under the two-sample summary 
statistics setting. We evaluate the performance of \texttt{spaceTSIV} 
with simulated data and compare it with existing (non-sparse) methods that 
work with two-sample summary statistics. 
Finally, we showcase our methods using proteomic and gene-expression data 
within the context of a drug-target discovery analysis.
Notation is summarized below and all proofs are provided in 
Supplementary Material~\ref{app:proofs}.

\begin{notation}
For all $k\in\mathbb{N}$, we define $[k]\coloneqq\{1,\ldots,k\}$ and for all 
$\beta\in\bbR^d$, we denote by $\operatorname{supp}(\beta) \coloneqq \left\{j \in [d]: \beta^j \neq 0 \right\}$ 
the set of non-zero components of $\beta$.
For an arbitrary matrix $A\in\mathbb{R}^{n\times m}$, we denote for all $i\in[n]$ and $j\in[m]$,
the $i$-th row of $A$ by $A_{i}$, the $j$-th column of $A$ by $A^{j}$, and the  $ij$-th entry of $A$ by $A_i^j$.
If $A$ is a square block matrix containing $k \times k$ square matrices of 
dimension $l\times l$, then $A^{[ij]}$ for all $i,j \in [k]$ denotes 
the $ij$-th block of $A$. 
\end{notation}

\section{Reduced form IV model and summary statistics} \label{sec:method}

We start from the conventional one-sample individual-level data setting 
and assume we observe $n$ independently and identically distributed (iid) 
observations $\{(X_i, Y_i, Z_i)\}_{i=1}^n\subseteq\mathbb{R}^d\times\mathbb{R}\times\mathbb{R}^m$, 
where $Y$ is a response variable, $X$ a vector of endogenous covariates, 
and $Z$ a vector of instruments. The IV model assumptions can then be 
expressed as a linear structural causal model (SCM) over these variables.\footnote{
The required assumptions can also be expressed via other causal models 
(e.g., potential outcomes). Not all causal implications of the model 
introduced here are strictly necessary, but to keep the presentation 
concise we avoid presenting the most general assumptions.} 
Formally, for all $i\in[n]$, we assume,
\begin{align} \label{eq:individual_scm}
\begin{split}
    X_i &\coloneqq AZ_i + BX_i + g(H_i, \nu_i^X) \\
    Y_i &\coloneqq X_i^\top \beta^* + h(H_i, \nu_i^Y),
\end{split}
\end{align}
where $H_i\in\bbR^q$ is a vector of unobserved variables, 
$g$ and $h$ are arbitrary measurable functions, 
$Z_i$, $h(H_i, \nu_i^X)$, and $g(H_i, \nu_i^Y)$ have mean $0$ and finite variance, 
and $\left\{Z_i, H_i, \nu_i^X, \nu_i^Y\right\}_{i=1}^n$ are jointly independent. 
The coefficient $\beta^*\in\bbR^d$ denotes the true causal effect of the covariates 
on the response, and the matrices $A\in\bbR^{d\times m}$ and $B\in\bbR^{d\times d}$ 
encode the other causal relations in the SCM, with $B$ being a strictly lower 
triangular matrix. 
The matrix $I_d - B$ is assumed to be invertible, where $I_d$ 
is the identity matrix of dimension $d$. 
Finally, we call the support of $\beta^*$ the parent set of $Y$ 
and denote it as $\text{PA}(Y)$, that is, 
$\text{PA}(Y) \coloneqq \operatorname{supp}(\beta^*)$. 
The SCM \eqref{eq:individual_scm} can also be expressed in what is called 
its reduced form by only considering how the instruments affect the covariates 
and the response. Formally, for all $i\in[n]$ the reduced form is given by
\begin{align} \label{eq:reduced_scm}
\begin{split}
    X_i &\coloneqq Z_i^\top\Pi + u_i^X \\
    Y_i &\coloneqq Z_i^\top\pi + u_i^Y,
\end{split}
\end{align}
where $\Pi \coloneqq A^\top (I_d - B)^{-\top} \in \bbR^{m \times d}$, 
$\pi \coloneqq \Pi\beta^* \in \bbR^m$, 
$u_i^X \coloneqq g(H_i,\nu_i^X)^\top (I_d - B)^{-\top}$, and 
$u_i^Y \coloneqq (u_i^X)^\top\beta^* + h(H_i,\nu_i)$.

In this work, we assume that we do not directly observe the individual-level 
data and instead only have access to summary statistics of partially 
observed paired data from two independent samples
$\{(Y_{ai}, Z_{ai})\}_{i=1}^{n_a}$ and $\{(X_{bi}, Z_{bi})\}_{i=1}^{n_b}$ of the 
SCM \eqref{eq:individual_scm}.

As discussed in Section~\ref{sec:intro}, this is often the  
case in MR studies utilizing summary statistics from two GWAS, one contains 
the associations between genetic variants and endogenous traits
(such as gene expression levels), and the other contains the 
associations between genetic variants and an outcome trait 
(such as a disease), are used to study the causal relationship between 
the endogenous traits and the outcome trait with genetic variants being 
the IVs.

There are two types of summary statistics that we focus on here. Firstly, 
the two-sample joint OLS summary statistics, which consist of estimates 
of the reduced form parameters in \eqref{eq:reduced_scm} and are formally 
defined as follows.

\begin{definition}[Two-sample joint OLS summary statistics] \label{def:joint_ols_sumstats}
Given two independent samples of observations 
$\{(Y_{ai}, Z_{ai})\}_{i=1}^{n_a}$ and $\{(X_{bi}, Z_{bi})\}_{i=1}^{n_b}$,
the \emph{two-sample joint OLS summary statistics} (joint summary statistics) are defined as the set of estimates 
\begin{equation*}
\cD^{\text{joint}}_{a,b} \coloneqq 
\left\{
\widehat\pi, \widehat\Sigma_\pi, 
\widehat\Pi, \widehat\Sigma_\Pi
\right\},
\end{equation*}
where 
$\widehat\pi \coloneqq (\bZ_a^\top\bZ_a)^{-1}\bZ_a^\top \bY_a \in\bbR^{m}$, 
$\widehat\Sigma_\pi \coloneqq \widehat\varepsilon_a^\top \widehat\varepsilon_a(\bZ_a^\top \bZ_a)^{-1} \in\bbR^{m\times m}$ 
with $\widehat\varepsilon_a \coloneqq \bY_a - \bZ_a\widehat\pi$, 
$\widehat\Pi \coloneqq (\bZ_b^\top\bZ_b)^{-1}\bZ_b^\top\bX_b \in\bbR^{m\times d}$,
and $\widehat\Sigma_\Pi\in\bbR^{md\times md}$ consists of $d\times d$ blocks of dimension $m\times m$ 
defined for all $k,l\in[d]$ by
$\widehat\Sigma_\Pi^{[kl]} \coloneqq (\widehat\varepsilon_b^k)^\top\widehat\varepsilon_b^l(\bZ_b^\top\bZ_b)^{-1}$ with
$\widehat\varepsilon_b^k \coloneqq \bX^k_b - \bZ_b \widehat\Pi^k$.
\end{definition}

Secondly, the two-sample marginal OLS summary statistics, which instead of capturing 
the joint effects described by the parameters in \eqref{eq:reduced_scm}, only 
contain marginal univariate effects.

\begin{definition}[Two-sample marginal OLS summary statistics] \label{def:marginal_ols_sumstats}
Given two independent samples of observations 
$\{(Y_{ai}, Z_{ai})\}_{i=1}^{n_a}$ and $\{(X_{bi}, Z_{bi})\}_{i=1}^{n_b}$,
the \emph{two-sample marginal OLS summary statistics} (marginal summary statistics) are defined as the set of estimates 
\begin{equation*}
\cD^{\text{marginal}}_{a, b} \coloneqq 
\left\{
\widehat  \eta, \widehat\sigma_ \eta^2, 
\widehat H, \widehat\sigma_H^2, \widehat{M}_{Z_a}, \widehat{M}_{Z_b}, \widehat{M}_X
\right\},
\end{equation*}
where $\widehat  \eta\in\bbR^m$, $\widehat\sigma_ \eta^2 \in \bbR^m$, 
$\widehat H \in\bbR^{m\times d}$, $\widehat\sigma_H^2 \in \bbR^{m\times d}$, 
and for all $j\in[m]$ and all $k\in[d]$, 
$\widehat  \eta_j \coloneqq (\bZ_a^j)^\top \bY_a / (\bZ_a^j)^\top \bZ_a^j$,
$\widehat\sigma_{ \eta,j}^2 \coloneqq (\widehat\varepsilon_a^j)^\top \widehat\varepsilon_a^j / \left((\bZ_a^j)^\top \bZ_a^j\right)$
with $\widehat\varepsilon_a^j \coloneqq \bY_a - \widehat  \eta_j \bZ_a^j$, 
$\widehat H_j^k \coloneqq (\bZ_b^j)^\top \bX_b^k / (\bZ_b^j)^\top \bZ_b^j$, and 
$(\widehat\sigma_{ H, j}^k)^2 \coloneqq (\widehat\varepsilon_{bj}^k)^\top \widehat\varepsilon_{bj}^k/ \left((\bZ_b^j)^\top \bZ_b^j\right)$
with $\widehat\varepsilon_{bj}^k \coloneqq \bX_b^k - \widehat H_j^k \bZ_b^j$. 
For both $s\in\{a,b\}$, 
let $D_{Z_s}$ be the diagonal matrix containing the diagonal elements of $\bZ_s^\top\bZ_s$, then 
$\widehat{M}_{Z_s} \coloneqq 
D_{Z_s}^{-1/2}\bZ_s^\top\bZ_s D_{Z_s}^{-1/2} \in\bbR^{m\times m}$ are the sample correlation 
matrices of $Z_s$ respectively. 
Similarly, let $D_{X}$ be the diagonal matrix containing the diagonal elements of $\bX_b^\top\bX_b$, 
then $\widehat{M}_X \coloneqq D_X^{-1/2}\bX_b^\top\bX_b D_X^{-1/2} \in\bbR^{d\times d}$
is the sample correlation matrix of $X_b$.
\end{definition}

Due to the close relation of the joint summary statistics with the reduced 
form model \eqref{eq:reduced_scm} it is easier to develop methods for the 
joint summary statistics. However, in most publicly available data 
(e.g., \citet{ukbb} and \citet{gwascatalog}) 
only the marginal summary statistics are available. Fortunately, it is 
possible to transform marginal summary statistics into joint summary statistics. 
This means that any theoretical developments that apply to one also apply to the other. 
The exact correspondence is given in the following proposition.

\begin{proposition}[Marginal to joint summary statistics] \label{prop:marginal_to_joint}
Assume we are given 
$\cD^{\text{marginal}}_{a, b}=\{\widehat  \eta, \widehat\sigma_ \eta^2, 
\widehat H, \widehat\sigma_ H^2,\widehat{M}_{Z_a}, \widehat{M}_{Z_b}, \widehat{M}_X
\}.$
Define diagonal matrices $D_a,D_b^{(1)},\ldots,D_b^{(m)}\in\mathbb{R}^{m\times m}$ such that 
for all $k,i\in[m]$, 
$(D_{a})_i^i\coloneqq(\widehat\sigma_{ \eta,i}^2 + (\widehat  \eta_i)^2)^{1/2}$
and 
$(D_{b}^{(k)})_i^i\coloneqq((\widehat\sigma_{H, i}^k)^2 + (\widehat H_i^k)^2)^{1/2}$.
Then it holds for all $k,l\in[d]$ that 

\noindent $\bullet$ $\widehat\pi = D_a(D_a\ \widehat{M}_{Z_a})^{-1}\ \widehat  \eta,$ \\
\noindent $\bullet$ $\widehat\Sigma_\pi = 
(1 - \widehat  \eta^\top D_a \widehat{M}_{Z_a}^{-1} D_a \widehat \eta) D_a\widehat{M}_{Z_a}^{-1}D_a$, \\
\noindent $\bullet$ $\widehat\Pi^k = D_b^{(k)}(D_b^{(k)}\ \widehat{M}_{Z_b})^{-1}\ \widehat H^k$, and \\
\noindent $\bullet$ $\widehat\Sigma_\Pi^{[kl]} = (\widehat{M}_{X,k}^l - \widehat H^{k \top} D_b^{(k)} \widehat{M}_{Z_b}^{-1} D_b^{(l)}\widehat H^l)D_b^{(k)}\widehat{M}_{Z_b}^{-1}D_b^{(l)}$. 
\end{proposition}

In practice, one often does not observe both $\widehat{M}_{Z_{a}}$ and $\widehat{M}_{Z_b}$ and instead only observes a single estimate that converges to the correlation of $Z$. In such cases, it can be shown that using the same transformation as in Proposition~\ref{prop:marginal_to_joint} is asymptotically equivalent to working with the joint summary statistics.

\subsection{Identifiability via sparsity under the reduced IV model}

For the causal effect $\beta^*$ to be identified,
the number of instruments is usually 
required to be no less than the number of covariates. In the one-sample 
individual-level data setting, this can be seen from the 
solution space based on the IV moment condition under the SCM \eqref{eq:individual_scm}, 
\begin{equation} \label{eq:id_solution_space}
\cB^{\text{ind}} = \left\{\beta\in\bbR^d: \bbE(Z\,Y) = \bbE(Z X^\top) \beta)\right\}.
\end{equation}
This space is in general non-degenerate if the dimension of the instruments 
is larger than the number of covariates.
When the causal effect is sparse, however, 
it is possible to allow more covariates than instruments. 

\begin{figure}[!thb]
\centering
\resizebox{0.3\textwidth}{!}{
\begin{tikzpicture}
    % Style definitions
    \tikzstyle{var}=[circle, draw, thick, minimum size=6mm, font=\small, inner sep=0]
    \tikzstyle{arrow}=[-latex, thick]
    \tikzstyle{doublearrow}=[latex-latex, thick]
    \tikzstyle{dashedarrow}=[-latex, thick, dashed]

    % Nodes for I_1 to I_J (G variables in sketch)
    % \node[var] (I4) at (1, 4) {$Z_4$};
    % \node[var] (I5) at (1, 0) {$Z_5$};
    \node[var] (I1) at (0, 3) {$Z_1^a$};
    \node[var] (I2) at (0, 2) {$Z_2^a$};
    \node[var] (I3) at (0, 1) {$Z_3^a$};

    % Nodes for X_1 to X_R
    \node[var, fill=gray!30] (X4) at (3, 2.5) {$X_4^a$};
    \node[var, fill=gray!30] (X5) at (3, 1.5) {$X_5^a$};
    \node[var, fill=gray!30] (X1) at (1.5, 3) {$X_1^a$};
    \node[var, fill=gray!30] (X2) at (1.5, 2) {$X_2^a$};
    \node[var, fill=gray!30] (X3) at (1.5, 1) {$X_3^a$};

    % Node for H
    % \node[var] (H) at (3.2, 3) {$H$};

    % Node for Y
    \node[var] (Y) at (4.5, 2) {$Y^a$};

    % % Double-arrow arcs between I variables (G variables in sketch)
    % \draw[doublearrow, bend right=45] (I1) to (I3);
    % \draw[doublearrow, bend left=45] (I2) to (I5);
    % \draw[doublearrow, bend right=45] (I3) to (I4);
    % \draw[doublearrow, bend right=45] (I4) to (I5);
    % \draw[doublearrow, bend right=45] (I3) to (I5);

    % % Arrows from some I variables to some X variables
    \draw[arrow] (I1) -- (X1);
    \draw[arrow] (I2) -- (X2);
    \draw[arrow] (I3) -- (X3);
    % \draw[arrow] (I4) -- (X4);
    % \draw[arrow] (I5) -- (X5);
    % \draw[dashedarrow, bend left] (I4) to (Y);
    % \draw[dashedarrow, bend right] (I5) to (Y);

    % % H points to both X1 and Y
    % \draw[arrow] (H) -- (X1);
    % \draw[arrow] (H) -- (X2);
    % \draw[arrow] (H) -- (XR);
    % \draw[arrow] (H) -- (Y);

    % Some X points to X
    \draw[arrow] (X1) -- (X4);
    \draw[arrow] (X2) -- (X4);
    \draw[arrow] (X2) -- (X5);
    \draw[arrow] (X3) -- (X5);

    % Some X points to Y
    \draw[arrow] (X4) -- (Y);
    \draw[arrow] (X5) -- (Y);

\end{tikzpicture}
}
\hspace{1cm}
\resizebox{0.3\textwidth}{!}{
\begin{tikzpicture}
    % Style definitions
    \tikzstyle{var}=[circle, draw, thick, minimum size=6mm, font=\small, inner sep=0]
    \tikzstyle{arrow}=[-latex, thick]
    \tikzstyle{doublearrow}=[latex-latex, thick]
    \tikzstyle{dashedarrow}=[-latex, thick, dashed]

    % Nodes for I_1 to I_J (G variables in sketch)
    % \node[var] (I4) at (1, 4) {$Z_4$};
    % \node[var] (I5) at (1, 0) {$Z_5$};
    \node[var] (I1) at (0, 3) {$Z_1^b$};
    \node[var] (I2) at (0, 2) {$Z_2^b$};
    \node[var] (I3) at (0, 1) {$Z_3^b$};

    % Nodes for X_1 to X_R
    \node[var] (X4) at (3, 2.5) {$X_4^b$};
    \node[var] (X5) at (3, 1.5) {$X_5^b$};
    \node[var] (X1) at (1.5, 3) {$X_1^b$};
    \node[var] (X2) at (1.5, 2) {$X_2^b$};
    \node[var] (X3) at (1.5, 1) {$X_3^b$};

    % Node for H
    % \node[var] (H) at (3.2, 3) {$H$};

    % Node for Y
    \node[var, fill=gray!30] (Y) at (4.5, 2) {$Y^b$};

    % % Double-arrow arcs between I variables (G variables in sketch)
    % \draw[doublearrow, bend right=45] (I1) to (I3);
    % \draw[doublearrow, bend left=45] (I2) to (I5);
    % \draw[doublearrow, bend right=45] (I3) to (I4);
    % \draw[doublearrow, bend right=45] (I4) to (I5);
    % \draw[doublearrow, bend right=45] (I3) to (I5);

    % % Arrows from some I variables to some X variables
    \draw[arrow] (I1) -- (X1);
    \draw[arrow] (I2) -- (X2);
    \draw[arrow] (I3) -- (X3);
    % \draw[arrow] (I4) -- (X4);
    % \draw[arrow] (I5) -- (X5);
    % \draw[dashedarrow, bend left] (I4) to (Y);
    % \draw[dashedarrow, bend right] (I5) to (Y);

    % % H points to both X1 and Y
    % \draw[arrow] (H) -- (X1);
    % \draw[arrow] (H) -- (X2);
    % \draw[arrow] (H) -- (XR);
    % \draw[arrow] (H) -- (Y);

    % Some X points to X
    \draw[arrow] (X1) -- (X4);
    \draw[arrow] (X2) -- (X4);
    \draw[arrow] (X2) -- (X5);
    \draw[arrow] (X3) -- (X5);

    % Some X points to Y
    \draw[arrow] (X4) -- (Y);
    \draw[arrow] (X5) -- (Y);

\end{tikzpicture}
}
\caption{An example of two-sample IV scenario that is considered as 
underidentified in the usual sense. Hidden confounders between $X$ 
and $Y$ are omitted for clarity. 
While the two DAGs have the same structure, in sample $a$ (left)
the covariates $X$ are not observed and in sample $b$ (right) the outcome $Y$ 
is not observed, these unobserved variables are represented by gray nodes.}
\label{fig:example_two_sample}
\end{figure}
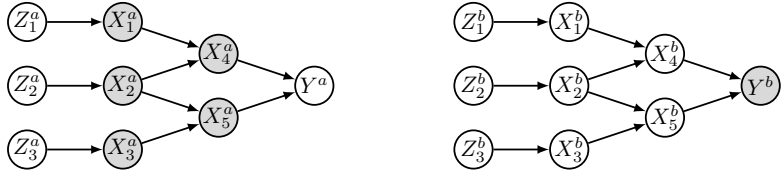

\citet{pfister2022identifiability} study in detail 
the identifiability conditions under the SCM \eqref{eq:individual_scm}. 
In the following, we describe the identifiability conditions under the 
reduced model \eqref{eq:reduced_scm} which is compatible with 
the two-sample summary statistics scenario. 
Lemma~\ref{lem:ts_solution_space} describes the solution space 
of the causal effects with the reduced form model. 
\begin{lemma} \label{lem:ts_solution_space}
If $\mathbb{E}[ZZ^{\top}]$ has full rank, the solution space of the causal 
effects based on the IV moment condition can be written as 
\begin{equation} \label{eq:ss_solution_space}
\cB^{\text{sum}} = \left\{\beta\in\bbR^d: \pi - \Pi\beta = 0\right\}.
\end{equation}
\end{lemma}

We will focus on the case where the instruments do not have direct effects 
on the response, which is implied by the SCMs \eqref{eq:individual_scm} and 
its reduced form \eqref{eq:reduced_scm}.
This is usually referred to as the exclusion restriction criteria of IV. 
In genetics research, such a direct effect is also referred to as 
pleiotropy \citep[see e.g.,][]{hemani2018evaluating}.
We demonstrate empirically with additional simulations 
in Supplementary Material~\ref{app:experiment} that the proposed methods still 
perform well under small violations of this assumption.
An example of the possible scenario represented by directed acyclic graphs (DAGs) 
is given in Figure~\ref{fig:example_two_sample}. As we will see 
shortly, although the number of instrument is less than the number 
of covariates in this case, the causal effect from $X$ to $Y$ may 
still be identified.

Under the two-sample summary statistic setting, the identifiability 
conditions in \citet{pfister2022identifiability} can be written as follows.\footnote{For 
a matrix $A$, $\text{Im}(A)$ denotes the image of $A$.}
\begin{assumption} \label{ass:sparse_identify}
For all $S\subseteq [d]$, let $\Pi^S$ be the submatrix of 
$\Pi$ containing the columns in $S$. We assume the following regarding 
the true parameter $\Pi$ 
\begin{enumerate}[label=(\alph*)]
\item $\text{rank}(\Pi^{\text{PA}(Y)}) = |\text{PA}(Y)|$. 
\item $\forall S\subseteq [d]$, it holds that 
$\text{rank}(\Pi^S) \leq \text{rank}(\Pi^{\text{PA}(Y)})$ and 
$\text{Im}(\Pi^S) \neq \text{Im}(\Pi^{\text{PA}(Y)})$ imply 
that $\forall w\in\bbR^{|S|}$, $\Pi^S w\neq \Pi^{\text{PA}(Y)} (\beta^*)^{\text{PA}(Y)}$.
\item $\forall S\subseteq [d]$ with $|S| = |\text{PA}(Y)|$ and 
$S\neq \text{PA}(Y)$, $\text{Im}(\Pi^S) \neq \text{Im}(\Pi^{\text{PA}(Y)})$.
\end{enumerate}
\end{assumption}

To obtain a sparse solution, it is natural to consider the following 
optimization problem 
\begin{equation} \label{eq:sparse_l0}
    \min_{\beta\in\cB^{\text{sum}}} ||\beta||_0.
\end{equation}
Theorem~\ref{thm:identify_truebeta} shows that under Assumption~\ref{ass:sparse_identify}, 
$\beta^*$ is a unique solution to \eqref{eq:sparse_l0}. 
The proof follows similarly as in \citet[][Theorem 3]{pfister2022identifiability}.

\begin{theorem}[Identifiability of sparse causal effect with reduced form model] \label{thm:identify_truebeta}
If Assumption~\ref{ass:sparse_identify} (a) and (b) hold, 
then $\beta^*$ is a solution to \eqref{eq:sparse_l0}. If in addition 
Assumption~\ref{ass:sparse_identify} (c) holds, then $\beta^*$ is the unique solution.
\end{theorem}

\subsection{Anderson-Rubin test for two-sample summary statistics}

The AR test is a well-known weak-instrument 
robust test for the causal effect, and the LIML estimator is known to minimize the AR statistic 
\citep[e.g.,][]{dhrymes2012econometrics}. 
\citet{wang2022weak} consider the two-sample summary statistic version of the 
AR test when there is a single covariate, which can be seen as a 
a generalization of the the modified Q statistic proposed by 
\citet{bowden2019improving} for independent instruments. 
The following result is a generalization of \citet{wang2022weak} 
in the presence of multiple covariates\footnote{A related result is also considered 
by \citet{patel2024robust} where a dispersion parameter is included
and the principal components of the instruments are used. Here we focus on the 
case where the instruments are valid.}, which will be referred to as 
the Q statistic. 

\begin{theorem}[Q statistic] \label{thm:q_stat}
Assume Assumption~\ref{ass:regularity} holds. 
For all $\beta \in \bbR^d$, define the Q statistic as
\begin{equation} \label{eq:gen_qstat}
Q(\beta) \coloneqq (\widehat{\pi}-\widehat{\Pi}\beta)^{\top} (\tfrac{1}{n_a}\widehat\Sigma_\pi + \tfrac{1}{n_b}\widehat\Sigma_\Pi(\beta))^{-1}(\widehat{\pi}-\widehat{\Pi}\beta),
\end{equation}
where $\widehat\Sigma_\Pi(\beta) \coloneqq \xi(\beta)\widehat\Sigma_\Pi\xi^\top(\beta)$ with $\xi(\beta) \coloneqq \beta^\top\otimes I_m$.
Then it holds for all $\beta\in\mathbb{R}^d$ and all $r\in(0,\infty)$ that
\begin{equation*}
    \lim_{\substack{n_a,n_b\rightarrow\infty\\n_a/n_b\rightarrow r}}\sup_{t\in\mathbb{R}}\sup_{\substack{P\in\mathcal{P}:\\\beta\in\cB^{\text{sum}}(P)}}\left|\mathbb{P}_P\left(Q(\beta)\leq t\right)-\kappa_m(t)\right| = 0,
\end{equation*}
where $\kappa_m$ is the CDF of the chi-squared distribution with $m$-degrees of freedom.
\end{theorem}

The Q statistic is the two-sample counterpart of the one-sample AR statistic, 
we present their connections in Supplementary Material~\ref{app:ar_vs_qstat}. 
Its minimizer can also be viewed as a generalized method of moments 
(GMM) estimator \citep{hansen1982large}, and it is related to the J~statistic in economics literature.
See Remark~\ref{rmk:sig_pi_beta} in the Supplementary Material for 
additional comments on the definition of $\widehat\Sigma_\Pi(\beta)$.

\section{Estimating sparse causal effects with spaceTSIV}

We describe two estimation procedures to the optimization problem \eqref{eq:sparse_l0}. 
The first procedure is the two-sample summary statistics counterpart of \texttt{spaceIV} by 
\citet{pfister2022identifiability}, and the second procedure employs L1-penalization to 
replace subset selection which has the advantage of faster computational speed. 
For both procedures, we will use the following estimator, which is the minimizer of the Q statistic 
constrained on a specific support. 
For all $S \subseteq\{1,\ldots,d\}$, define 
\begin{equation} \label{eq:beta_hat_subset}
\widehat\beta^Q(S) \coloneqq \argmin_{\substack{\beta\in\bbR^d: \text{ supp}(\beta) = S}} Q(\beta).
\end{equation}
In order to provide precise theoretical results, we
further let $\mathcal{P}$ denote a family of distributions for 
$(X, Y, Z)$ generated by \eqref{eq:individual_scm} which is assumed to be
sufficiently regular (see Assumption~\ref{ass:regularity} for details). 
For all $P \in \cP$, we let $\beta^*(P)$ denote the causal 
effect and $\cB^\text{sum}(P)$ be the subset 
$\cB^\text{sum}$ induced by the distribution $P$ (both of which are fully 
identified from the observational distribution $P$).

\subsection{Sparsity by subset selection} \label{sec:subset_selection}

For all $s\in[d]$, let  
\begin{equation*}
\widehat\beta^Q(s) \coloneqq \widehat\beta^Q\Big(\textstyle\argmin_{S\subseteq\{1,\ldots, d\}:|S| = s} Q\left(\widehat\beta^Q(S)\right) \Big).
\end{equation*}
Moreover, following Theorem~\ref{thm:q_stat}, for all $s\in[d]$ and for all $\alpha\in(0,1)$, 
the hypothesis test 
\begin{equation*}
\varphi^{\alpha}_s(\cD^{\text{joint}}_{a, b})
\coloneqq \mathbbm{1}\left(Q(\widehat\beta^Q(s)) > \kappa_m^{-1}(1-\alpha)\right)
\end{equation*}
has uniform asymptotic level for the null hypothesis
$$\cH_0(s)\coloneqq\{P\in\mathcal{P}\mid \exists \beta\in\cB^{\text{sum}}(P):\, \norm{\beta}_0 = s\},$$
that is, for $\alpha\in(0,1)$, it holds that 
$$\lim_{n_a,n_b\to\infty} \sup_{P\in\cH_0(s)}\bbP_{P}(\varphi^{\alpha}_s(\cD^{\text{joint}}_{a, b}) = 1) \leq \alpha.$$
An algorithm defining the \texttt{spaceTSIV} estimator 
using subset selection is given in Algorithm~\ref{alg:spacetsiv_l0}. 
Theorem~\ref{thm:consistency_subset} shows that it is consistent.

\begin{theorem} \label{thm:consistency_subset}
Assume Assumption~\ref{ass:regularity} holds. Let $\cD^{\text{joint}}_{a,b}$ be the joint summary statistics 
based on two independent samples of size $n_a$ and $n_b$ respectively. 
Let $P\in\mathcal{P}$ and $s_{\text{max}} \in \bbN$ such that $s_{\text{max}}\geq \norm{\beta^*(P)}_0$. 
If Assumption~\ref{ass:sparse_identify} (a) and (b) holds, then for all $r\in(0,\infty)$ 
\begin{equation*}
\lim_{\substack{n_a,n_b\rightarrow\infty\\n_a/n_b\rightarrow r}} \mathbb{P}_P\left(\norm{\widehat\beta^{\leq s_\text{max}}}_0 = \norm{\beta^*}_0\right) \geq 1-\alpha;
\end{equation*}
if in addition Assumption~\ref{ass:sparse_identify} (c) also holds, then 
for all $\varepsilon > 0$ and all $r\in(0,\infty)$
\begin{equation*}
\lim_{\substack{n_a,n_b\rightarrow\infty\\n_a/n_b\rightarrow r}} \mathbb{P}_P\left(\norm{\widehat\beta^{\leq s_\text{max}} - \beta^*}_2 < \varepsilon \right) \geq 1-\alpha.
\end{equation*}
\end{theorem}

\begin{algorithm}
\caption{\texttt{spaceTSIV} with L0 penalization} \label{alg:spacetsiv_l0}
\DontPrintSemicolon
\SetKwInput{KwIn}{Input}
\SetKwInput{KwOut}{Output}
\KwIn{Joint summary statistics $\cD^{\text{joint}}_{a,b}$, 
maximum support size $s_\text{max}$, significance level $\alpha\in(0,1)$} 
Initialize $s \gets 1$ and $\phi \gets 1$ \;
\While{$s \leq s_\text{max}$ and $\phi = 1$}{
Set $\mathbf{S}_s$ to the set of all subsets of $[d]$ of size $s$ \;
\For{$S \in \mathbf{S}_s$}{
Compute $\widehat\beta^Q(S)$ \;
Compute $Q(\widehat\beta^Q(S))$ \;
}
$S_\text{best} \gets \argmin_{S\in\mathbf{S}_s} Q(\widehat\beta^Q(S))$ \; % \tcp{Best set among sets of size $s$}
$\widehat\beta(s) \gets \widehat\beta^Q(S_\text{best})$ \;
$\phi \gets \varphi_{s}^\alpha(\cD^{\text{joint}}_{a,b})$ \;
$s \gets s + 1$ 
}
$\widehat\beta_{\leq s_{\text{max}}} \gets \widehat\beta(s)$ \;
\KwOut{Final estimate $\widehat\beta_{\leq s_{\text{max}}}$ and test result $\phi$}
\end{algorithm}

\subsection{Sparsity by L1 penalty} \label{sec:method_l1}

The subset selection approach introduced in Section~\ref{sec:subset_selection} 
becomes computationally infeasible when the number of covariates is large. 
We therefore propose a faster approach that uses L1 penalization to estimate 
the support of $\beta^*$ and then adapt the testing procedure from the previous 
section. More specifically, for a penalty parameter $\lambda>0$, we first minimize 
the following L1-loss
\begin{equation} \label{eq:tsiv_lasso}
\cL_\lambda^{\text{TSIV-L1}}(\beta) = \frac{1}{2} \norm{\widehat\pi - \widehat\Pi\beta}_2^2 + \lambda\norm{\beta}_1.
\end{equation}
Define $\widehat\beta(\lambda) \coloneqq \argmin_{\beta\in\mathbb{R}^d} \cL_\lambda^{\text{TSIV-L1}}(\beta)$ 
and $\widehat{S}_\lambda \coloneqq \operatorname{supp}(\widehat\beta(\lambda))$. 
We then propose to refit the parameter as in \eqref{eq:beta_hat_subset} using 
the set $\widehat{S}_{\lambda}$ and performing the hypothesis test defined by
\begin{equation*}
    \varphi_\lambda^\alpha(\cD^{\text{joint}}_{n_a, n_b}) \coloneqq  \mathbbm{1}\left(Q(\widehat\beta^Q(\widehat{S}_\lambda)) > \kappa_m^{-1}(1-\alpha)\right).
\end{equation*}
By similar arguments as in Section~\ref{sec:subset_selection} this test for 
$S=\widehat{S}_\lambda$ has uniform asymptotic level for the null hypothesis
\begin{equation*}
    H_0(S)\coloneqq\{P\in\mathcal{P}\mid \exists \beta \in \cB^{\text{sum}}(P): \operatorname{supp}(\beta)=S\}.
\end{equation*}
Under sufficient regularity conditions and assuming that $\beta^*$ is indeed 
sparse, one can hope---based on similar results for high-dimensional linear 
models \citep[e.g.,\ ][]{buhlmann2011statistics}---that for appropriately chosen 
$\lambda$ it holds that $\widehat{S}_{\lambda}$ converges to $\operatorname{supp}(\beta^*)$. 
This motivates the following estimator, 
\texttt{spaceTSIV} with L1 penalization,  
defined in Algorithm~\ref{alg:spacetsiv_l1}.

\begin{algorithm}
\caption{\texttt{spaceTSIV} with L1 penalization} \label{alg:spacetsiv_l1}
\DontPrintSemicolon
\SetKwInput{KwIn}{Input}
\SetKwInput{KwOut}{Output}
\KwIn{Joint summary statistics $\cD^{\text{joint}}_{a,b}$, 
a vector of penalty values in decreasing order $\{\lambda_1,\ldots,\lambda_\ell\}$, 
significance level $\alpha\in(0,1)$} 
Initialize $l \gets 1$ and $\phi \gets 1$\;
\While{$l \leq \ell$ and $\phi = 1$}{
$\lambda \gets \lambda_l$ \;
$\widehat S_\lambda \gets \text{supp}\left(\argmin_{\beta\in\mathbb{R}^d} \cL_{\lambda}^{\text{TSIV-L1}}(\beta)\right)$ \;
Compute $\widehat\beta^Q(\widehat S_\lambda)$\;
$\phi \gets \varphi_{\lambda}^\alpha(\cD^{\text{joint}}_{a,b})$ \;
$l \gets l+1$ \;
}
$\widehat\beta_{\leq \lambda_{\text{max}}} \gets \widehat\beta(\widehat{S}_{\lambda})$ \; 
\KwOut{Final estimate $\widehat\beta_{\leq \lambda_{\text{max}}}$ and test result $\phi$}
\end{algorithm}

Intuitively, if the subset selection is indeed correct (i.e., it recovers 
the support of $\beta^*$) for the first accepted set, then this procedure 
should correctly estimate $\beta^*$. A full theoretical analysis, however, 
goes beyond the scope of this work and we propose this procedure only as 
a heuristic computational speed up.

\subsection{Practical considerations} \label{sec:practicality}

When using the subset selection approach in practice, it can happen that 
there are multiple estimates with different support of the same (smallest) 
size not being rejected by $\varphi_s^\alpha$. This indicates, that at least 
in finite sample, the causal effect $\beta^*$ is not fully identified.
We recommend reporting all subsets of the smallest size that are not 
rejected by $\varphi_s^\alpha$ as possible effects.

Moreover, since the estimator \texttt{spaceTSIV} is based on optimizing 
a test statistics, one immediate approach to construct confidence intervals (CIs) 
is by inverting the test. In the real application, we construct the CIs 
for the non-zero causal effects by inverting $\varphi_s^\alpha$ or $\varphi_\lambda^\alpha$ and projecting onto 
each non-zero coordinate. We choose this approach for its practicality, but 
other approaches exist which may be more suitable \citep[e.g.,][]{londschien2024weak}, 
and one should also take into account the effect of post-selection inference 
\citep[e.g.,][]{lee2016exact}. 
One of the advantages of inverting the test is that is takes into account 
the strength of the instruments (and hence identifiability). So if the 
resulting CIs are unbounded this generally indicates that there is limited 
identifiability.
This is well known property for the AR test \citep[e.g.,][]{dufour1997some,davidson2014confidence}.

%%%%%%%%%%%%%%%%%%%%%%%%%%%%%%%%%%%%%%%%%%%%%%%%%%%%%%%%%%%%%%%%%%%%%%%%%%%%%%%%
\section{Experiments} \label{sec:experiment}
%-------------------------------------------------------------------------------

Code for reproducing the simulations and the real-data application  
along with the data are available in the GitHub repository \url{https://github.com/shimenghuang/spacetsiv}.
All experiments were run on a MacBook Pro laptop with M1 chip. 

\subsection{Simulations} \label{sec:simulation}

We present simulation results for two data generating processes (DGPs) 
summarized below in this section. Further simulation results are provided 
in Supplementary Material~\ref{app:experiment}.
The first, DGP1, is a low dimensional example taken from 
\citet[][Figure~3]{pfister2022identifiability}.
We compare the subset selection and the L1-penalization versions of 
\texttt{spaceTSIV}, denoted as \texttt{spaceTSIV-L0} and \texttt{spaceTSIV-L1}
respectively, as well as the \texttt{TSIV} estimator (defined as the minimizer of 
\eqref{eq:tsiv_lasso} with $\lambda = 0$, in which case the generalized inverse is used). 
The second, DGP2, illustrates the scenario with higher dimensional covariates, 
sparser causal effects, and correlated instruments. 
In this setting we omit \texttt{spaceTSIV-L0}
from the comparison due its high computational cost. 
Overview of the simulation setup is given below
and more details can be found in Supplementary Material~\ref{app:experiment_detail}.

\noindent\textbf{DGP1 overview:} $m = 3$ and $d = 5$ and $\norm{\beta^*}_0 = 2$. For increasing 
$n =n_a = n_b$, we generate iid $\{(Y_i, Z_i)\}_{i=1}^{n_a}$ and 
$\{(X_i, Z_i)\}_{i=1}^{n_b}$ according to a linear SCM with Gaussian errors
and then compute the summary statistics using seemingly unrelated regression.  \\ 
\noindent\textbf{DGP2 overview:} $m = 5$, $d= 100$, and $\norm{\beta^*}_0 = 2$. 
With fixed values of $\pi$, $\Pi$, 
$\Sigma_\pi$, and $\Sigma_\Pi$,
and increasing $n = n_a = n_b$, we generate 
$\widehat\pi_{n_a} \sim \cN(\pi, \frac{1}{n_a}\Sigma_\pi)$ and
$\widehat\Pi_{n_b} \sim \cN(\Pi, \frac{1}{n_b}\Sigma_\Pi)$, and set 
$\widehat\Sigma_{\pi,n_a} = \Sigma_{\pi}$
and $\widehat\Sigma_{\Pi,n_b} = \Sigma_{\Pi}$. 

\begin{figure}[!thb]
\centering
\includegraphics[width=\linewidth]{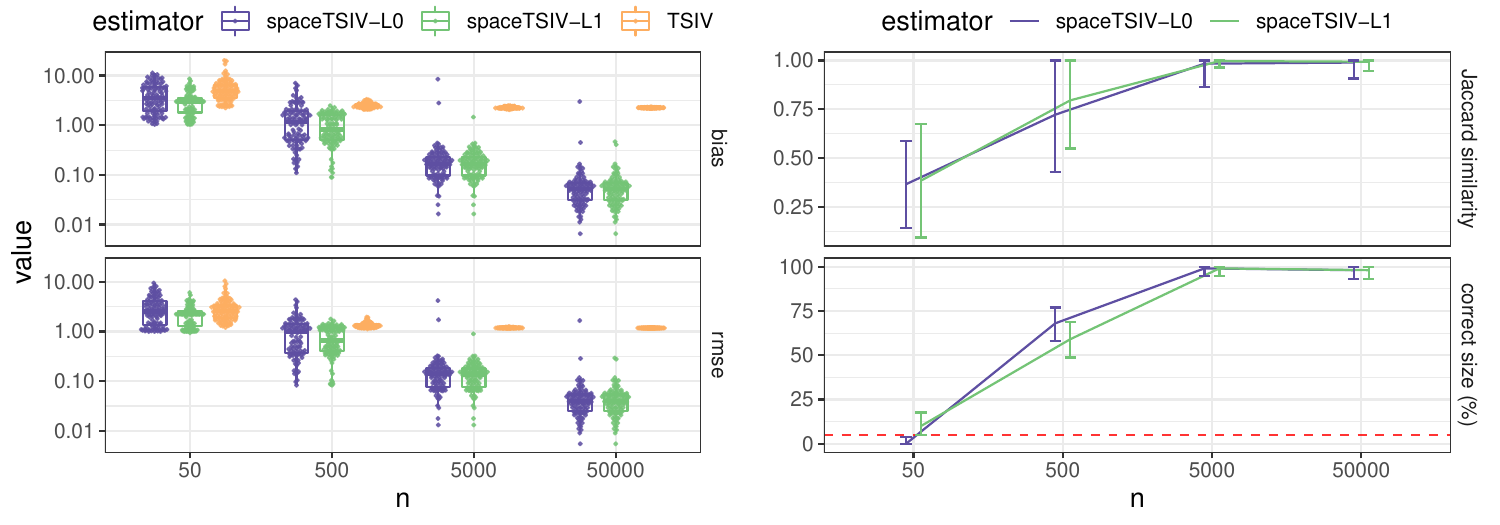}
\hspace{5pt}
\includegraphics[width=\linewidth]{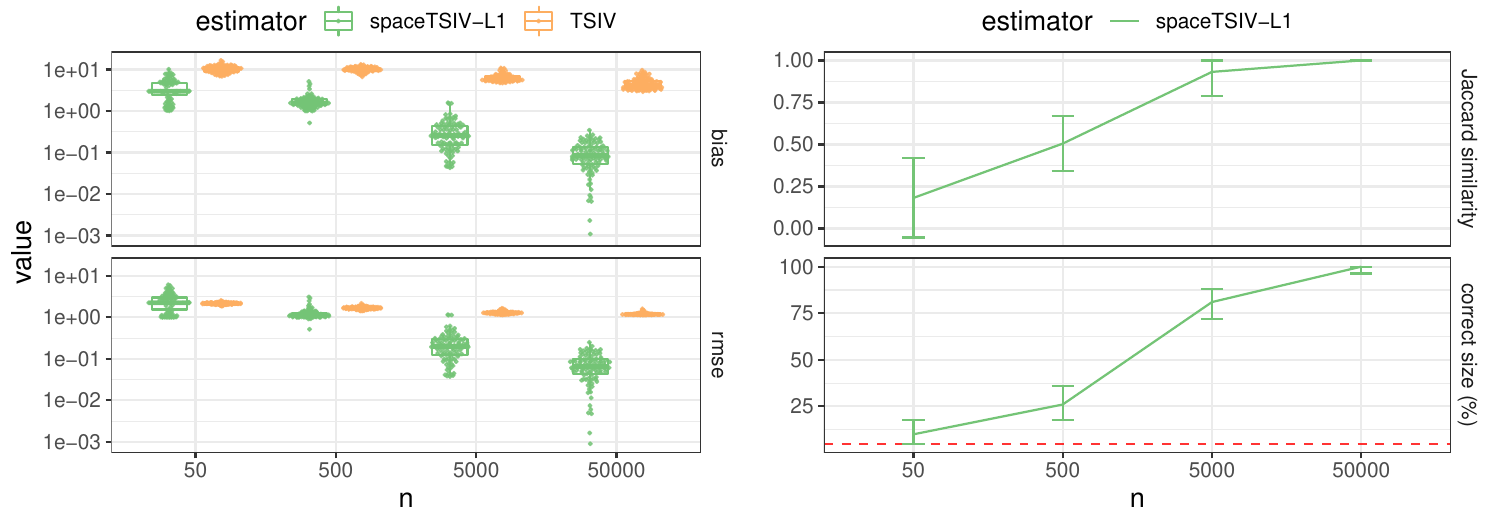}
\caption{Results using data generated by DGP1 (top) and DGP2 (bottom) 
based on $100$ repetitions. 
Left: Bias and rmse of the estimators. The y-axis is on log scale for clarity. 
Right: Average Jaccard similarity between the selected covariates and 
the true causal covariates (error bars indicate confidence intervals 
constructed by mean plus/minus one standard error), 
and percentage of estimates that have the correct support size 
(error bars indicate $95\%$ binomial confidence intervals).}
\label{fig:dgp_eval}
\end{figure}

We evaluate \texttt{spaceTSIV} based on both its variable selection and 
estimation performances. 
The results are shown in Figure~\ref{fig:dgp_eval}. 
We can see that the bias and rmse of \texttt{spaceTSIV} shrinks with 
increasing sample size with either L0 or L1 penalization, which is not the 
case for the non-sparse estimator \texttt{TSIV}. 
In terms of variable selection, we see that for both DGPs as the sample 
sizes increase, the Jaccard similarity\footnote{
For two sets $A$ and $B$, the Jaccard similarity is defined as 
$\text{Jaccard}(A, B) \coloneqq \displaystyle\frac{|A \cap B|}{|A\cup B|}$.}
increases to around $1$, and the percentage of estimates having 
the correct support size also increases to around $100\%$, 
empirically confirming
the consistency results in Theorem~\ref{thm:consistency_subset}. 
The performance of \texttt{spaceTSIV-L0} and \texttt{spaceTSIV-L1} are 
similar in terms of both estimation and variable selection for DGP1. 

\subsection{Application} \label{sec:application}

We apply our methods to summary statistics of SNP-level associations where 
the covariates and the outcome come from two separate GWAS sources. The 
covariates' summary statistics come from the GTEX consortium, which measure 
levels of expression of protein coding genes 
across multiple tissue types in the human body. Gene expression is a convenient 
and reliable upstream marker of protein production, which would be the 
natural target of a future drug. We specifically focus on expression of the GLP1R gene in $10$ 
tissue types that are relevant to the treatment of cardio-metabolic disease.
These are brain caudate, hypothalamus, atrial appendage, left ventricle, lung, 
nerve, pancreas, stomach, testis, and thyroid. The SNP-outcome summary statistics 
measure the genetic association with  coronary artery disease (CAD) risk, 
and are obtained from the CARDIoGRAMplusC4D consortium. These data were first 
analysed in in \citet{patel2024robust}, who proposed a novel principle 
component analysis (PCA) method for constructing orthogonal composite 
instruments from $851$ SNPs in the GLP1R gene region. For this analysis, 
they use $23$ principle components (PCs) as IVs for the $10$ covariates. 
The analysis by \citet{patel2024robust} suggests that GLP1R expression only 
has a significant effect on CAD risk in $2$ of the $10$ tissues, although this was based on 
$95\%$ confidence intervals using a normal approximation which, 
unlike the test-inversion method we use, does not always reliably capture the 
true uncertainty of IV estimates when the instruments are weak.

Based on the analysis of \citet{patel2024robust}, it is reasonable to 
believe that the causal effects are sparse in this application. Rather than opting 
for PCA pre-processing of the genetic summary statistics, we consider the selection 
of individual SNPs instruments based on the more conventional approach using the 
first-stage F-statistics\footnote{
Given a marginal OLS coefficient $\widehat\gamma \in\bbR$ and its 
corresponding standard error $\widehat\sigma \in\bbR$, the first-stage 
F-statistic is defined as $\widehat\gamma^2 / \widehat\sigma^2$.} of the 
gene expression summary statistics. 
We keep the top two genetic variants with the largest first-stage F-statistics 
for each of the $10$ covariates. Since some SNPs are most strongly associated 
with multiple covariates, we eventually keep $17$ of the $851$ genetic variants 
in the original data. 
Moreover, since the summary statistic data contains only the marginal 
associations along with their standard errors, we use the adjustment method 
in Proposition~\ref{prop:marginal_to_joint} to obtain the estimated joint 
effects and variance-covariance matrices. 

\begin{figure}[!tbh]
\centering
\includegraphics[width=0.9\linewidth]{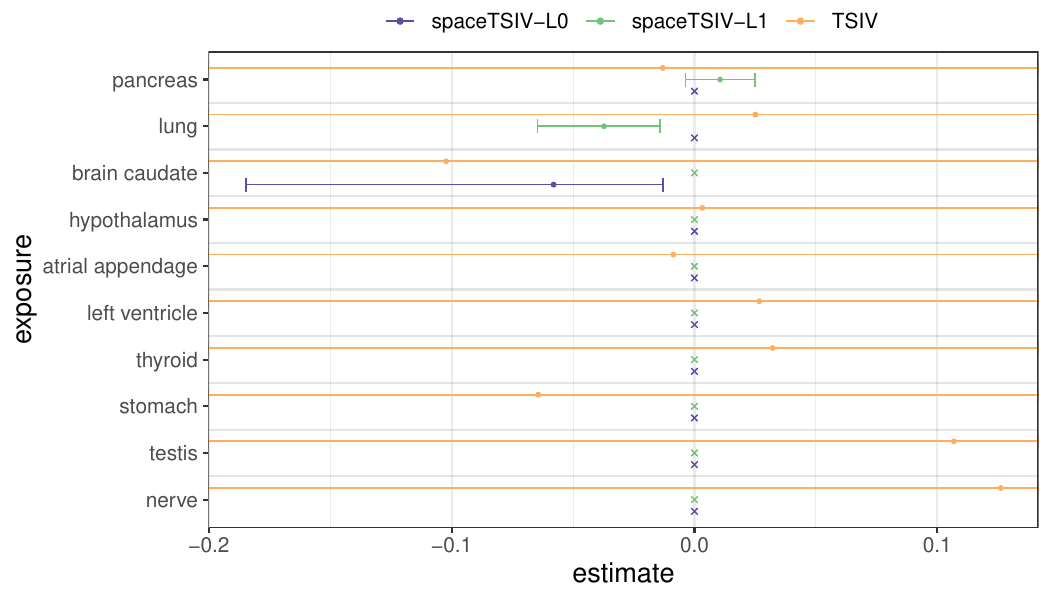}
\caption{Estimated effects of the GLP1R expression in $10$ tissues using the 
selected $17$ genetic variants as instruments. 
Error bars represent $90\%$ confidence intervals (CIs) constructed by inverting $\varphi_s^\alpha$ and $\varphi_\lambda^\alpha$ respectively, and projecting onto each coordinate. 
}
\label{fig:glp1r_main}
\end{figure}

The analysis results based on \texttt{spaceTSIV} with L0 and L1 penalization 
and regular \texttt{TSIV} are reported in Figure~\ref{fig:glp1r_main}, where the 90\% CIs
are obtained from inverting $\varphi_s^\alpha$ and $\varphi_\lambda^\alpha$ as described in Section~\ref{sec:practicality}. 
They show that the CIs for \texttt{TSIV} are all of infinite length. 
This demonstrates that, even though there are more instruments than covariates, the causal 
effects are still under-identified due to weak instruments. 
Moreover, the \texttt{spaceTSIV} with L0 penalization yields a single 
set of size $1$ while with L1 penalization we obtain a set of size $2$. 
The significant negative effect of brain caudate aligns with the analysis 
result in \citet{patel2024robust} and is biologically meaningful. 
The different result from \texttt{spaceTSIV-L1} could be due to the high correlation 
of the SNPs, which may result in the L1 relaxation of the L0 minimization 
problem not achieving the same estimate. 
In general we recommend using the L0 procedure whenever computationally feasible as it comes with clear theoretical guarantees.

%%%%%%%%%%%%%%%%%%%%%%%%%%%%%%%%%%%%%%%%%%%%%%%%%%%%%%%%%%%%%%%%%%%%%%%%%%%%%%%%
\section{Discussion} \label{sec:disuss}
%-------------------------------------------------------------------------------

We propose \texttt{spaceTSIV} for sparse multivariable causal effect estimation 
under unobserved confounding, which is applicable to the two-sample summary 
statistics setting. Two methods using subset selection and L1-penalization respectively 
are provided. We prove consistency for the subset selection approach and illustrate 
the results in simulations. We also show in simulations that the L1-penalization 
approach, which is much more computational efficient, can achieve similar 
performance as the subset selection approach in terms of bias and consistency. 
To focus on the main idea of this work, we have assumed that the summary 
statistics utilized in the analysis are obtained from two independent and 
homogeneous samples, which is commonly assumed in genetic epidemiology. 
However, it would be interesting to generalize the methods 
to heterogeneous samples similar to results by \citet{zhao2019two} in the non-sparse setting.
Moreover, if the summary statistics are obtained from two samples with 
overlapping observations, additional correlations should be taken into account. 

\section*{Acknowledgement}

The authors would like to thank Stephen Burgess and Ashish Patel for helpful
discussions at the start of this research project, and Anton Rask Lundborg for
helpful discussions on the uniform asymptotic results. This work was partially
completed during SH's research visit at Novo Nordisk. The authors would like to
thank Jesper Ferkinghoff-Borg, Kang Li and Lewis Marsh for facilitating this
visit and for discussing necessary concepts and tools in statistical genetics at
an early stage. SH and NP are supported by a research grant (0069071) from Novo
Nordisk Fonden. JB is funded at the University of Exeter by research grant
MR/X011372/1.

\clearpage

\vskip 0.2in

\bibliographystyle{abbrvnat}
\bibliography{bibliography} 

\begin{thebibliography}{5}
\providecommand{\natexlab}[1]{#1}
\providecommand{\url}[1]{\texttt{#1}}
\expandafter\ifx\csname urlstyle\endcsname\relax
  \providecommand{\doi}[1]{doi: #1}\else
  \providecommand{\doi}{doi: \begingroup \urlstyle{rm}\Url}\fi

\bibitem[Anderson and Rubin(1949)]{anderson1949estimation}
T.~W. Anderson and H.~Rubin.
\newblock Estimation of the parameters of a single equation in a complete
  system of stochastic equations.
\newblock \emph{The Annals of Mathematical Statistics}, 20\penalty0
  (1):\penalty0 46--63, 1949.

\bibitem[Johnson and Horn(1985)]{johnson1985matrix}
C.~R. Johnson and R.~A. Horn.
\newblock \emph{Matrix analysis}.
\newblock Cambridge university press Cambridge, 1985.
\newblock \doi{10.1017/CBO9780511810817}.

\bibitem[Klyne and Shah(2023)]{klyne2023average}
H.~Klyne and R.~D. Shah.
\newblock Average partial effect estimation using double machine learning.
\newblock \emph{arXiv preprint arXiv:2308.09207}, 2023.

\bibitem[Londschien and B{\"u}hlmann(2024)]{londschien2024weak}
M.~Londschien and P.~B{\"u}hlmann.
\newblock Weak-instrument-robust subvector inference in instrumental variables
  regression: A subvector lagrange multiplier test and properties of subvector
  {A}nderson-{R}ubin confidence sets.
\newblock \emph{arXiv preprint arXiv:2407.15256}, 2024.

\bibitem[Lundborg et~al.(2022)Lundborg, Kim, Shah, and
  Samworth]{lundborg2022projected}
A.~R. Lundborg, I.~Kim, R.~D. Shah, and R.~J. Samworth.
\newblock The projected covariance measure for assumption-lean variable
  significance testing.
\newblock \emph{arXiv preprint arXiv:2304.01098}, 2022.

\end{thebibliography}


\begin{thebibliography}{23}
\providecommand{\natexlab}[1]{#1}
\providecommand{\url}[1]{\texttt{#1}}
\expandafter\ifx\csname urlstyle\endcsname\relax
  \providecommand{\doi}[1]{doi: #1}\else
  \providecommand{\doi}{doi: \begingroup \urlstyle{rm}\Url}\fi

\bibitem[Anderson and Rubin(1949)]{anderson1949estimation}
T.~W. Anderson and H.~Rubin.
\newblock Estimation of the parameters of a single equation in a complete
  system of stochastic equations.
\newblock \emph{The Annals of Mathematical Statistics}, 20\penalty0
  (1):\penalty0 46--63, 1949.

\bibitem[Batool et~al.(2022)Batool, Patel, Gill, and
  Burgess]{batool2022disentangling}
F.~Batool, A.~Patel, D.~Gill, and S.~Burgess.
\newblock Disentangling the effects of traits with shared clustered genetic
  predictors using multivariable {M}endelian randomization.
\newblock \emph{Genetic Epidemiology}, 46\penalty0 (7):\penalty0 415--429,
  2022.

\bibitem[Bowden et~al.(2019)Bowden, Del Greco~M, Minelli, Zhao, Lawlor,
  Sheehan, Thompson, and Davey~Smith]{bowden2019improving}
J.~Bowden, F.~Del Greco~M, C.~Minelli, Q.~Zhao, D.~A. Lawlor, N.~A. Sheehan,
  J.~Thompson, and G.~Davey~Smith.
\newblock Improving the accuracy of two-sample summary-data {M}endelian
  randomization: moving beyond the nome assumption.
\newblock \emph{International Journal of Epidemiology}, 48\penalty0
  (3):\penalty0 728--742, 2019.

\bibitem[B{\"u}hlmann and Van De~Geer(2011)]{buhlmann2011statistics}
P.~B{\"u}hlmann and S.~Van De~Geer.
\newblock \emph{Statistics for high-dimensional data: methods, theory and
  applications}.
\newblock Springer, 2011.

\bibitem[Davidson and MacKinnon(2014)]{davidson2014confidence}
R.~Davidson and J.~G. MacKinnon.
\newblock Confidence sets based on inverting {A}nderson--{R}ubin tests.
\newblock \emph{The Econometrics Journal}, 17\penalty0 (2):\penalty0 S39--S58,
  2014.

\bibitem[Dhrymes(2012)]{dhrymes2012econometrics}
P.~J. Dhrymes.
\newblock \emph{Econometrics: Statistical foundations and applications}.
\newblock Springer, 2012.

\bibitem[Dufour(1997)]{dufour1997some}
J.-M. Dufour.
\newblock Some impossibility theorems in econometrics with applications to
  structural and dynamic models.
\newblock \emph{Econometrica: Journal of the Econometric Society}, 65\penalty0
  (6):\penalty0 1365--1387, 1997.

\bibitem[Grant and Burgess(2021)]{grant2021pleiotropy}
A.~J. Grant and S.~Burgess.
\newblock Pleiotropy robust methods for multivariable {Mendelian}
  randomization.
\newblock \emph{Statistics in Medicine}, 40\penalty0 (26):\penalty0 5813--5830,
  2021.

\bibitem[Grant and Burgess(2022)]{grant2022efficient}
A.~J. Grant and S.~Burgess.
\newblock An efficient and robust approach to mendelian randomization with
  measured pleiotropic effects in a high-dimensional setting.
\newblock \emph{Biostatistics}, 23\penalty0 (2):\penalty0 609--625, 2022.

\bibitem[{GWAS Catalog}()]{gwascatalog}
{GWAS Catalog}.
\newblock \newline URL \url{https://www.ebi.ac.uk/gwas/}, accessed 2024-10-02.

\bibitem[Hansen(1982)]{hansen1982large}
L.~P. Hansen.
\newblock Large sample properties of generalized method of moments estimators.
\newblock \emph{Econometrica: Journal of the Econometric Society}, 4\penalty0
  (4):\penalty0 1029--1054, 1982.

\bibitem[Hemani et~al.(2018)Hemani, Bowden, and
  Davey~Smith]{hemani2018evaluating}
G.~Hemani, J.~Bowden, and G.~Davey~Smith.
\newblock Evaluating the potential role of pleiotropy in {M}endelian
  randomization studies.
\newblock \emph{Human Molecular Genetics}, 27\penalty0 (R2):\penalty0
  R195--R208, 2018.

\bibitem[Lee et~al.(2016)Lee, Sun, Sun, and Taylor]{lee2016exact}
J.~D. Lee, D.~L. Sun, Y.~Sun, and J.~E. Taylor.
\newblock Exact post-selection inference, with application to the lasso.
\newblock \emph{The Annals of Statistics}, 44\penalty0 (3):\penalty0 907--927,
  2016.

\bibitem[Londschien and B{\"u}hlmann(2024)]{londschien2024weak}
M.~Londschien and P.~B{\"u}hlmann.
\newblock Weak-instrument-robust subvector inference in instrumental variables
  regression: A subvector lagrange multiplier test and properties of subvector
  {A}nderson-{R}ubin confidence sets.
\newblock \emph{arXiv preprint arXiv:2407.15256}, 2024.

\bibitem[Patel et~al.(2024)Patel, Gill, Shungin, Mantzoros, Knudsen, Bowden,
  and Burgess]{patel2024robust}
A.~Patel, D.~Gill, D.~Shungin, C.~S. Mantzoros, L.~B. Knudsen, J.~Bowden, and
  S.~Burgess.
\newblock Robust use of phenotypic heterogeneity at drug target genes for
  mechanistic insights: Application of cis-multivariable {M}endelian
  randomization to {GLP1R} gene region.
\newblock \emph{Genetic Epidemiology}, 48\penalty0 (4):\penalty0 151--163,
  2024.

\bibitem[Pfister and Peters(2022)]{pfister2022identifiability}
N.~Pfister and J.~Peters.
\newblock Identifiability of sparse causal effects using instrumental
  variables.
\newblock In J.~Cussens and K.~Zhang, editors, \emph{Proceedings of the
  Thirty-Eighth Conference on Uncertainty in Artificial Intelligence}, volume
  180 of \emph{Proceedings of Machine Learning Research}, pages 1613--1622.
  PMLR, 2022.

\bibitem[Rees et~al.(2019)Rees, Wood, Dudbridge, and Burgess]{rees2019robust}
J.~M. Rees, A.~M. Wood, F.~Dudbridge, and S.~Burgess.
\newblock Robust methods in {M}endelian randomization via penalization of
  heterogeneous causal estimates.
\newblock \emph{PLOS One}, 14\penalty0 (9):\penalty0 e0222362, 2019.

\bibitem[Tang et~al.(2023)Tang, Kong, and Wang]{tang2023synthetic}
D.~Tang, D.~Kong, and L.~Wang.
\newblock The synthetic instrument: From sparse association to sparse
  causation.
\newblock \emph{arXiv preprint arXiv:2304.01098}, 2023.

\bibitem[{UK Biobank}()]{ukbb}
{UK Biobank}.
\newblock URL \url{https://www.ukbiobank.ac.uk/}, accessed 2024-10-02.

\bibitem[Wang et~al.(2021)Wang, Zhao, Bowden, Hemani, Davey~Smith, Small, and
  Zhang]{jingshu2020}
J.~Wang, Q.~Zhao, J.~Bowden, G.~Hemani, G.~Davey~Smith, D.~S. Small, and N.~R.
  Zhang.
\newblock Causal inference for heritable phenotypic risk factors using
  heterogeneous genetic instruments.
\newblock \emph{PLOS Genetics}, 17\penalty0 (6):\penalty0 1--24, 06 2021.

\bibitem[Wang and Kang(2022)]{wang2022weak}
S.~Wang and H.~Kang.
\newblock Weak-instrument robust tests in two-sample summary-data {M}endelian
  randomization.
\newblock \emph{Biometrics}, 78\penalty0 (4):\penalty0 1699--1713, 2022.

\bibitem[Zhao et~al.(2019)Zhao, Wang, Spiller, Bowden, and Small]{zhao2019two}
Q.~Zhao, J.~Wang, W.~Spiller, J.~Bowden, and D.~S. Small.
\newblock Two-sample instrumental variable analyses using heterogeneous
  samples.
\newblock \emph{Statistical Science}, 34\penalty0 (2):\penalty0 317--333, 2019.

\bibitem[Zhao et~al.(2020)Zhao, Wang, Hemani, Bowden, and
  Small]{zhao2020statistical}
Q.~Zhao, J.~Wang, G.~Hemani, J.~Bowden, and D.~S. Small.
\newblock Statistical inference in two-sample summary-data {M}endelian
  randomization using robust adjusted profile score.
\newblock \emph{The Annals of Statistics}, 48\penalty0 (3):\penalty0
  1742--1769, 2020.

\end{thebibliography}

\clearpage

\appendix
\renewcommand{\thesection}{\Alph{section}}
\counterwithin{figure}{section}
\renewcommand\thefigure{\thesection\arabic{figure}}
\counterwithin{table}{section}
\renewcommand\thetable{\thesection\arabic{table}}

\raggedbottom
\allowdisplaybreaks

\section*{\centering Supplementary Material}

%%%%%%%%%%%%%%%%%%%%%%%%%%%%%%%%%%%%%%%%%%%%%%%%%%%%%%%%%%%%%%%%%%%%%%%%%%%%%%%%
\section{Details of test statistics and test-based estimators}
%-------------------------------------------------------------------------------

\subsection{Connection between Anderson-Rubin statistic and Q statistic} \label{app:ar_vs_qstat}

Suppose we observe one set of iid samples $\{(X_i, Y_i, Z_i)\}_{i=1}^{n}$. 
The AR statistic is given by 
\begin{equation} \label{eq:ar_individual}
AR(\beta) \coloneqq \frac{n-m}{m}\cdot\frac{(\bY-\bX\beta)^\top P_Z(\bY-\bX\beta)}{(\bY-\bX\beta)^\top M_Z(\bY-\bX\beta)},
\end{equation}
where $P_Z = \bZ(\bZ^\top \bZ)^{-1}\bZ^\top$ and $M_Z = I_d - P_Z$. 
When the true causal effect $\beta^*$ is identified, $m\cdot AR(\beta^*)\dto\chi^2_m$ \citepAPP{anderson1949estimation, londschien2024weak}. 

We can rewrite the AR statistic in terms of (joint) OLS estimates and their 
respective estimated variance-covariance matrices. Specifically, let 
$\widehat\pi = (\bZ^\top \bZ)^{-1}\bZ\bY$ and 
$\widehat\Pi = (\bZ^\top \bZ)^{-1}\bZ\bX$,  
we have that 
\begin{equation*}
(\bY-\bX\beta)^\top P_Z(\bY-\bX\beta) = (\widehat\pi - \widehat\Pi\beta)^\top(\bZ^\top\bZ)(\widehat\pi - \widehat\Pi\beta).
\end{equation*}
Moreover, for all
$\beta\in\mathbb{R}^d$ define
\begin{align*}
\widehat\Sigma_\pi &= (\bY - \bZ\widehat\pi)^\top(\bY - \bZ\widehat\pi)(\bZ^\top\bZ)^{-1}, \\
\widehat\Sigma_\Pi(\beta) &= \beta^{\top}(\bX - \bZ\widehat\Pi)^\top(\bX - \bZ\widehat\Pi)\beta(\bZ^\top\bZ)^{-1}, \text{ and } \\
\widehat\Sigma_{\pi,\Pi}(\beta) &= (\bY - \bZ\widehat\pi)^\top(\bX - \bZ\widehat\Pi)\beta(\bZ^\top\bZ)^{-1}. 
\end{align*}
Then, we can expand the denominator in \eqref{eq:ar_individual} as follows
\begin{align*}
    &(\bY-\bX\beta)^{\top}M_Z(\bY-\bX\beta)\\
    &\quad=(M_Z\bY-M_Z\bX\beta)^{\top}(M_Z\bY-M_Z\bX\beta)\\
    &\quad=(\bY-\bZ\widehat\pi-(\bX-\bX\widehat\Pi\beta))^{\top}(\bY-\bZ\widehat\pi-(\bX-\bX\widehat\Pi)\beta)\\
    &\quad= (\bY-\bZ\widehat\pi)^{\top}(\bY-\bZ\widehat\pi)+\beta^{\top}(\bX-\bX\widehat\Pi)^{\top}(\bX-\bX\widehat\Pi)\beta -2(\bY-\bZ\widehat\pi)^{\top}(\bX-\bX\widehat\Pi)\beta,
\end{align*}
which implies
\begin{equation*}
    (\bY-\bX\beta)^\top M_Z(\bY-\bX\beta)(\bZ^\top\bZ)^{-1} = 
\widehat\Sigma_\pi + \widehat\Sigma_\Pi(\beta)  - 2 \widehat\Sigma_{\pi,\Pi}(\beta).
\end{equation*}
Therefore, 
\begin{align*}
AR(\beta) &= \frac{(\widehat\pi - \widehat\Pi\beta)^\top(\bZ^\top\bZ)(\widehat\pi - \widehat\Pi\beta)}{(\bY-\bX\beta)^\top M_Z(\bY-\bX\beta)} \\
&= \frac{1}{m}(\widehat\pi - \widehat\Pi\beta)^\top
\left(\frac{1}{n-m}\widehat\Sigma_\pi + \frac{1}{n-m}\widehat\Sigma_\Pi(\beta)  - \frac{2}{n-m}\widehat\Sigma_{\pi,\Pi}(\beta)\right)^{-1}
(\widehat\pi - \widehat\Pi\beta).
\end{align*}
From this expression, we can see the connection
between the AR statistic and the Q statistic. Specifically, for a fixed $m$ and large $n$, the difference between 
$mAR(\beta)$ and $Q(\beta)$ is the term $\widehat\Sigma_{\pi,\Pi}$, which is related to the covariance between the residuals of a $Y$ on $Z$ and a $X$ on $Z$ regression. In the two sample setting this covariance is zero because the two regressions are performed on independent samples and therefore it is not needed in the the Q statistic.

\subsection{Coordinate descent for minimizing the TSIV-L1 loss}

We describe the coordinate descent procedure for minimizing \eqref{eq:tsiv_lasso}. 
Let $$\cL^{\text{TSIV}}(\beta) = \frac{1}{2}\|\widehat\pi - \widehat\Pi\beta\|_2^2.$$
For a matrix $A$, denote $A^{-j}$ as the matrix removing A's $j$-th column. 
The derivative of $\cL^{\text{TSIV}}(\beta)$ w.r.t $\beta_j$ is 
\begin{align} \label{eq:liv}
\begin{split}
\frac{\partial \cL^{\text{TSIV}}(\beta)}{\beta_j} &= -(\widehat\Pi^j)^\top(\widehat\pi - \widehat\Pi\beta) \\
&= -(\widehat\Pi^j)^\top\widehat\pi + (\widehat\Pi^j)^\top\widehat\Pi^{-j}\beta_{-j} + (\widehat\Pi^j)^\top\widehat\Pi^j\beta_j  \\
&= -\rho_j + \eta_j\beta_j
\end{split}
\end{align}
where $\rho_j\coloneqq (\widehat\Pi^j)^\top\widehat\pi + (\widehat\Pi^j)^\top\widehat\Pi^{-j}\beta_{-j}$ and 
$\eta_j \coloneqq (\widehat\Pi^j)^\top\widehat\Pi^j$. 
The subgradient of $\lambda||\beta||_1$ w.r.t $\beta_j$ is 
\begin{equation} \label{eq:liv_l1}
\frac{\partial \lambda||\beta||_1}{\beta_j} = \frac{\partial \lambda|\beta_j|}{\beta_j}
\begin{cases}
\{-\lambda\} & \beta_j < 0 \\
[-\lambda, \lambda] & \beta_j = 0 \\
\{-\lambda\} & \beta_j > 0
\end{cases}
\end{equation}
Combining \eqref{eq:liv} and \eqref{eq:liv_l1}, we have that the subgradient of 
$\cL_\lambda^{\text{TSIV-L1}}(\beta)$ w.r.t.\ $\beta_j$ is 
\begin{equation} \label{eq:liv_subgrad}
\frac{\partial \cL_\lambda^{\text{TSIV}}(\beta)}{\partial \beta_j} = 
\begin{cases}
-\rho^j + \eta^j\beta_j - \lambda & \beta_j < 0 \\
[-\rho^j -\lambda, -\rho^j +\lambda] & \beta_j = 0 \\
-\rho^j + \eta^j\beta_j + \lambda & \beta_j > 0.
\end{cases}
\end{equation}
Starting from an initial value of $\widehat\beta$, we loop through 
$j\in[J]$ and update the value of $\widehat\beta_j$ by solving the equation resulting from setting 
\eqref{eq:liv_subgrad} to $0$, which gives
\begin{equation*}
\widehat\beta_j = 
\begin{cases}
\frac{\rho^j + \lambda}{\eta^j} & \rho^j < -\lambda \\
0 & -\lambda < \rho^j < \lambda \\
\frac{\rho^j - \lambda}{\eta^j} & \rho^j > \lambda.
\end{cases}
\end{equation*}

\section{Regularity conditions}

\begin{assumption}[Regularity conditions]
\label{ass:regularity}
Let $\mathcal{P}$ be a family of distributions for $(X, Y, Z)\in\mathbb{R}^d\times\mathbb{R}\times\mathbb{R}^m$ generated by \eqref{eq:individual_scm} and additionally satisfies that there exists $C_1,C_2,c,\eta>0$ such that
\begin{itemize}
    \item $\sup_{P\in\mathcal{P}}\left(\mathbb{E}_P[\|X\|^{4+\eta}]+\mathbb{E}[\|Y\|^{4+\eta}]+\mathbb{E}_P[\|Z\|^{4+\eta}]\right)\leq c$
    \item $\inf_{P\in\mathcal{P}}\min\left(\lambda_{\min}(\mathbb{E}_P[ZZ^{\top}]),\lambda_{\min}(\mathbb{E}_P[XX^{\top}])\right)\geq C_1.$
    \item $\sup_{P\in\mathcal{P}}\max\left(\lambda_{\max}(\mathbb{E}_P[ZZ^{\top}]),\lambda_{\max}(\mathbb{E}_P[XX^{\top}])\right)\leq C_2.$
\end{itemize}
\end{assumption}

\section{Proofs} \label{app:proofs}

\subsection{Proof of Lemma~\ref{lem:ts_solution_space}}
\begin{proof}
The following equivalences hold
\begin{align*}
    \beta \in \cB^{\text{ind}} &\Longleftrightarrow \bbE(Z\, Y) = \bbE(Z X^\top)\beta \\
    &\Longleftrightarrow \bbE(ZZ^\top)\pi = \bbE(ZZ^\top)\Pi\beta \quad \text{since following \eqref{eq:reduced_scm},
    we have $\bbE(Z\, Y) = \bbE(ZZ^\top)\pi$} \\
    &\qquad \qquad \qquad \qquad \qquad \qquad \qquad \qquad \text{and $\bbE(Z X^\top)\beta = \bbE(ZZ^\top)\Pi\beta$} \\ 
    &\Longleftrightarrow \pi = \Pi\beta \qquad \text{since $\bbE[ZZ^\top]$ is full rank}.
\end{align*}
\end{proof}

\subsection{Proof of Proposition~\ref{prop:marginal_to_joint}}
\begin{proof}
We only prove the result for $\widehat\Pi$ and $\widehat\Sigma_\Pi$. $\widehat\pi$ and $\widehat\Sigma_\pi$
can be viewed as a special case of the former, with $d = 1$. 

We first express $D^{(k)}_b$ in terms of the design matrices. To this end, observe that for all $j \in [m]$ and all $k\in [d]$, we have from the marginal OLS summary statistics that 
\begin{equation*}
\begin{split}
\left(\widehat\sigma_{\eta,j}^k\right)^2 &= \frac{(\bX_b^k - \widehat{H}_j^k \bZ_b^j)^\top(\bX_b^k - \widehat{H}_j^k \bZ_b^j)}{(\bZ_b^j)^\top \bZ_b^j} \\
&= \frac{(\bX_b^k)^\top\bX_b^k - 2 \widehat{H}_j^k (\bX_b^k)^\top \bZ_b^j + (\widehat{H}_j^k)^2 (\bZ_b^j)^\top \bZ_b^j}{(\bZ_b^j)^\top \bZ_b^j} \\
&= \frac{(\bX_b^k)^\top\bX_b^k - 2 (\widehat{H}_j^k)^2(\bZ_b^j)^\top\bZ_b^j + (\widehat{H}_j^k)^2 (\bZ_b^j)^\top \bZ_b^j}{(\bZ_b^j)^\top \bZ_b^j} \\
&= \frac{(\bX_b^k)^\top\bX_b^k - (\widehat{H}_j^k)^2 (\bZ_b^j)^\top \bZ_b^j}{(\bZ_b^j)^\top \bZ_b^j} \\
&= \frac{(\bX_b^k)^\top\bX_b^k}{ (\bZ_b^j)^\top \bZ_b^j} - (\widehat{H}_j^k)^2.
\end{split}
\end{equation*}
This further implies that $(\bZ_b^j)^\top \bZ_b^j = \displaystyle\frac{(\bX_b^k)^\top\bX_b^k}{\left(\widehat\sigma_{\eta,j}^k\right)^2 + (\widehat{H}_j^k)^2}$, 
and hence 
$\displaystyle\left(\left(\widehat\sigma_{\eta,j}^k\right)^2 + (\widehat{H}_j^k)^2\right)^{-1}\widehat{H}_j^k
= \frac{(\bZ_b^j)^\top \bZ_b^j}{(\bX_b^k)^\top\bX_b^k}\widehat{H}_j^k 
= \frac{(\bX_b^k)^\top\bZ_b^j}{(\bX_b^k)^\top\bX_b^k}$. 
Therefore, if we define the diagonal matrix $D_{Z_b}$ for all $i\in[m]$ by $(D_{Z_b})_i^i\coloneqq(\bZ_b^i)^\top \bZ_b^i$, it holds that
\begin{equation*}
    D_b^{(k)} = D_{Z_b}^{-1/2}\left((\bX_b^k)^\top\bX_b^k\right)^{1/2}.
\end{equation*}
Using this result, for all $k\in[d]$, we can expand the joint OLS estimate $\widehat{\Pi}^k$ as follows
\begin{align*}
    \widehat\Pi^k & = \left(\bZ_b^\top \bZ_b\right)^{-1}\bZ_b^\top\bX_b^k \\
    &= \left(\bZ_b^\top \bZ_b\right)^{-1} D_{Z_b}\widehat{H}^k \\
    &= D_{Z_b}^{-1/2} \widehat M_{Z_b}^{-1} D_{Z_b}^{-1/2} D_{Z_b} \widehat{H}^k \\
    &= D_{Z_b}^{-1/2} \widehat M_{Z_b}^{-1} D_{Z_b}^{1/2} \widehat{H}^k \\
    &= D_b^{(k)}\left(\widehat M_{Z_b}D_b^{(k)}\right)^{-1} \widehat{H}^k
\end{align*}
Similarly, for all $k,l\in[d]$, the variance-covariance matrix between $\widehat\Pi^k$ and $\widehat\Pi^l$, can be expanded as follows.
\begin{align*}
&\widehat\Sigma_\Pi^{[kl]}\\
&= \left(\bX_b^k - \bZ_b\widehat\Pi^k\right)^\top\left(\bX_b^l - \bZ_b\widehat\Pi^l\right) (\bZ_b^\top \bZ_b)^{-1} \\
&=  \left((\bX_b^k)^\top\bX_b^l - (\bX_b^k)^\top\bZ_b\widehat\Pi^l - (\widehat\Pi^k)^{\top}(\bZ_b)^{\top}\bX_b^l + (\widehat\Pi^k)^\top\bZ_b^\top \bZ_b\widehat\Pi^l\right)(\bZ_b^\top \bZ_b)^{-1} \\
&=  \left((\bX_b^k)^\top\bX_b^l - 2(\widehat\Pi^k)^\top\bZ_b^\top \bZ_b\widehat\Pi^l + (\widehat\Pi^k)^\top\bZ_b^\top \bZ_b\widehat\Pi^l\right)(\bZ_b^\top \bZ_b)^{-1} \\ 
&=  \left((\bX_b^k)^\top\bX_b^l - (\widehat\Pi^k)^\top\bZ_b^\top \bZ_b\widehat\Pi^l\right)(\bZ_b^\top \bZ_b)^{-1} \\
&=  \left(\widehat{M}_{X,k}^l - 
\frac{\bZ_b^\top \bX_b^k}{(\bX_k)^\top \bX_b^k} 
\left(\frac{\bZ_b^\top\bZ_b}{\sqrt{(\bX_k)^\top \bX_b^k} \sqrt{(\bX_b^l)^\top \bX_b^l}}\right)^{-1}
\frac{\bZ_b^\top \bX_b^k}{(\bX_b^l)^\top \bX_b^l}
\right) \left(\frac{\bZ_b^\top\bZ_b}{\sqrt{(\bX_b^k)^\top \bX_b^k} \sqrt{(\bX_b^l)^\top \bX_b^l}}\right)^{-1} \\
&= \left(
\widehat{M}_{X,k}^l - (\widehat{H}^k)^\top D_b^{(k)} \widehat{M}_{Z_b}^{-1} D_b^{(l)}\widehat{H}^l 
\right) D_b^{(k)}\widehat{M}_{Z_b}^{-1}D_b^{(l)}.
\end{align*}
\end{proof}

\subsection{Proof of Theorem~\ref{thm:identify_truebeta}}
\begin{proof}
\textit{(First statement)} Assume Assumption~\ref{ass:sparse_identify} (a) and (b) hold. 
We would like to show that 
$$\beta^* \in \argmin_{\beta\in\cB^{\text{sum}}} \norm{\beta}_0.$$
Since $\beta^* \in\cB^{\text{sum}}$, it suffices to show that for all $\tilde{\beta}\in\cB^{\text{sum}}$, 
we have $\norm{\tilde{\beta}}_0 \geq |\text{PA}(Y)|$. Fix a $\tilde{\beta}\in\cB^{\text{sum}}$. 
Since $\tilde{\beta} \in\cB^{\text{sum}}$, it holds that 
$\pi = \Pi\tilde{\beta} = \Pi\beta^*$. 
Let $S \coloneqq \text{supp}(\tilde{\beta})$. 
Since $\forall j\in [d]\setminus \text{PA}(Y)$, $(\beta^*)^j = 0$, 
$\Pi\tilde{\beta} = \Pi\beta^*$ implies that 
\begin{equation} \label{eq:gamma_subset_eq}
    \Pi^S\tilde{\beta}^S = \Pi^{\text{PA}(Y)}(\beta^*)^{\text{PA}(Y)}.
\end{equation}
For the sake of contradiction, suppose that $|S| < |\text{PA}(Y)|$. 
Then by Assumption~\ref{ass:sparse_identify} (a), we have that 
\begin{equation*}
\text{rank}(\Pi^{\text{PA}(Y)}) = \text{dim}\left(\text{Im}(\Pi^{\text{PA}(Y)}\right) 
= |\text{PA}(Y)| > |S| \geq \text{dim}\left(\text{Im}(\Pi^S)\right) = \text{rank}(\Pi^S). 
\end{equation*}
This gives $\text{rank}(\Pi^{\text{PA}(Y)}) > \text{rank}(\Pi^S)$ which 
implies $\text{Im}(\Pi^{\text{PA}(Y)}) \neq \text{Im}(\Pi^S)$.
Then by Assumption~\ref{ass:sparse_identify} (b), we have that $\forall w\in\bbR^{|S|}$, 
$\Pi^S w \neq \Pi^{\text{PA}(Y)}(\beta^{*})^{\text{PA}(Y)}$, but this contradicts \eqref{eq:gamma_subset_eq}. 
This concludes the proof of the first statement. 

\textit{(Second statement)} It remains to show that there is no other solutions 
than $\beta^*$ when Assumption~\ref{ass:sparse_identify} (c) holds. 
Suppose for the sake of contradiction that there exists $\tilde{\beta}\in\cB^{\text{sum}}$ with 
$S\coloneqq \text{supp}(\tilde{\beta})= |\text{PA}(Y)|$ and $S \neq \text{PA}(Y)$. 
Similarly as above, since $\tilde{\beta} \in\cB^{\text{sum}}$, \eqref{eq:gamma_subset_eq} holds. 
Then by Assumption~\ref{ass:sparse_identify} (c) we have 
$\text{Im}(\Pi^S) \neq \text{Im}(\text{PA}(Y))$.
Moreover, by Assumption~\ref{ass:sparse_identify} (a) it holds that 
$$\text{rank}(\Pi^{\text{PA}(Y)}) = |\text{PA}(Y)| = |S| \geq \text{rank}(\Pi^S).$$
Therefore, by Assumption~\ref{ass:sparse_identify} (b) 
$\forall w\in\bbR^{|S|}$, $\Pi^S w \neq \Pi^{\text{PA}(Y)}(\beta^*)^{\text{PA}(Y)}$, 
which again contradicts \eqref{eq:gamma_subset_eq}. This concludes the proof 
of the second statement.
\end{proof}

\subsection{Proof of Theorem~\ref{thm:q_stat}}
\begin{proof}
First, observe that using $S_{n_a,n_b}$ as defined in Lemma~\ref{lem:s_beta_asymptotic}, we can express the Q statistic for all $\beta\in\mathbb{R}^d$ as
\begin{equation*}
    Q(\beta)=S_{n_a,n_b}(\beta)^{\top}S_{n_a,n_b}(\beta).
\end{equation*}
Moreover, for all $\beta\in\cB^{\text{sum}}$ it holds by definition that $\mu_{n_a,n_b}=0$, hence Lemma~\ref{lem:s_beta_asymptotic} implies that $S_{n_a,n_b}(\beta)$ converges uniformly to a standard Gaussian distribution as $n_a,n_b$ tend to infinity and $n_a/n_b\rightarrow r$ for $r\in(0,\infty)$. Hence, by the continuous mapping theorem it holds that
\begin{equation*}
\lim_{\substack{n_a,n_b\rightarrow\infty\\n_a/n_b\rightarrow r}}\sup_{\substack{P\in\mathcal{P}:\\\beta\in\cB^{\text{sum}}(P)}}\sup_{t\in\mathbb{R}}\left|\mathbb{P}_P(Q(\beta)\leq t)-\kappa_m(t)\right|=0,
\end{equation*}
which completes the proof of Theorem~\ref{thm:q_stat}.
\end{proof}

\subsection{Proof of Theorem~\ref{thm:consistency_subset}}

\begin{proof}
Let $r\in (0,\infty)$ and assume that $n_a / n_b\rightarrow r$ throughout the proof. 
Using $S_{n_a,n_b}$ as defined in Lemma~\ref{lem:s_beta_asymptotic}, we can express the Q statistic for all $\beta\in\mathbb{R}^d$ as
\begin{equation*}
    Q(\beta)=S_{n_a,n_b}(\beta)^{\top}S_{n_a,n_b}(\beta).
\end{equation*}
Furthermore, let $\overline{\cB}\subseteq\mathbb{R}^d$ be a compact set and choose $\overline{\beta}\in\overline{\cB}$ such that $\inf_{\beta\in\overline{\cB}}\|S_{n_a,n_b}(\beta)\|_2^2=\|S_{n_a,n_b}(\beta^*)\|_2^2$. Then, using standard probability bounds and dropping the $n_a,n_b$ from the notation for simplicity, we get for all $P\in\mathcal{P}$ and all $t\in[0,\infty)$ that
\begin{align}
    &\mathbb{P}_P\left(\inf_{\beta\in\overline{\cB}}\|S(\beta)\|_2^2\leq t\right)\\
    &=\mathbb{P}_P\left(\|S(\overline{\beta})-\mu(\overline{\beta})+\mu(\overline{\beta})\|_2\leq \sqrt{t}\right)\nonumber\\
    &\leq\mathbb{P}_P\left(\big|\|S(\overline{\beta})-\mu(\overline{\beta})\|_2-\|\mu(\overline{\beta})\|_2\big|\leq \sqrt{t}\right)\nonumber\\
    &=\mathbb{P}_P\left(\|S(\overline{\beta})-\mu(\overline{\beta})\|_2-\|\mu(\overline{\beta})\|_2\leq \sqrt{t}, \|S(\overline{\beta})-\mu(\overline{\beta})\|\geq\|\mu(\overline{\beta})\|\right)\nonumber\\
    &\qquad+\mathbb{P}_P\left(\|\mu(\overline{\beta})\|_2-\|S(\overline{\beta})-\mu(\overline{\beta})\|_2\leq \sqrt{t}, \|S(\overline{\beta})-\mu(\overline{\beta})\|\leq\|\mu(\overline{\beta})\|\right)\nonumber\\
    &\leq\mathbb{P}_P\left(\|S(\overline{\beta})-\mu(\overline{\beta})\|\geq\|\mu(\overline{\beta})\|\right)\nonumber\\
    &\qquad+
    \mathbb{P}_P\left(\|\mu(\overline{\beta})\|_2-\|S(\overline{\beta})-\mu(\overline{\beta})\|_2\leq \sqrt{t}\right)\nonumber\\
    &\leq 2\mathbb{P}_P\left(\|S(\overline{\beta})-\mu(\overline{\beta})\|_2\geq \|\mu(\overline{\beta})\|_2 - \sqrt{t}\right)\nonumber\\
    &\leq 2\mathbb{P}_P\left(\sup_{\beta\in\overline{\cB}}\|S(\beta)-\mu(\beta)\|_2\geq \inf_{\beta\in\overline{\cB}}\|\mu(\beta)\|_2 - \sqrt{t}\right).\label{eq:upper_bound_infQ}
\end{align}
Next, observe that
\begin{align*}
    S(\beta)
    &=\sqrt{n_b}\left(\tfrac{n_b}{n_a}\widehat{\Sigma}_\pi + \beta^{\top}\widehat{\Sigma}_X\beta\widehat{\Sigma}_{Z_b}^{-1}\right)^{-1/2}\left(\pi-\Pi\beta\right)\\
    &=\sqrt{n_b}\left(\tfrac{n_b}{n_a}\widehat{\Sigma}_\pi + (\beta/\|\beta\|_2)^{\top}\widehat{\Sigma}_X(\beta/\|\beta\|_2)\widehat{\Sigma}_{Z_b}^{-1}\right)^{-1/2}\left(\pi-\Pi(\beta/\|\beta\|_2)\right),
\end{align*}
where $\widehat{\Sigma}_{X}\coloneqq\frac{1}{n_b}\sum_{i=1}^{n_b}(X_{bi}-\widehat{\Pi}^{\top}Z_{bi})(X_{bi}-\widehat{\Pi}^{\top}Z_{bi})^{\top}$. This in particular implies that $S$ and hence $Q$ does not depend on the norm of $\beta$. Moreover, for all $\beta\in\mathbb{R}^d$ with $\|\beta\|_2=1$ it holds that
\begin{align*}
    \|S(\beta)-\mu(\beta)\|_2
    &=\sqrt{n_b}\|(\tfrac{n_b}{n_a}\widehat{\Sigma}_\pi + \beta^{\top}\widehat{\Sigma}_X\beta\widehat{\Sigma}_{Z_b}^{-1})^{-1/2}((\pi-\Pi\beta)-(\widehat\pi-\widehat\Pi\beta))\|_2\\
    &\leq \sqrt{n_b}\|\tfrac{n_b}{n_a}\widehat{\Sigma}_\pi + \beta^{\top}\widehat{\Sigma}_X\beta\widehat{\Sigma}_{Z_b}^{-1}\|_{\operatorname{op}}^{-1/2}(\|\pi-\widehat{\pi}\|_2 + \|\Pi\beta-\widehat\Pi\beta\|_2)\\
    &\leq \left(\lambda_{\min}(\tfrac{n_b}{n_a}\widehat{\Sigma}_\pi) + \lambda_{\max}(\beta^{\top}\widehat{\Sigma}_X\beta\widehat{\Sigma}_{Z_b}^{-1})\right)^{-1/2}(\sqrt{n_b}\|\pi-\widehat{\pi}\|_2 + \sqrt{n_b}\|\Pi-\widehat\Pi\|_{\operatorname{op}})\\
    &\leq \left(\lambda_{\min}(\widehat{\Sigma}_X)\lambda_{\max}(\widehat{\Sigma}_{Z_b}^{-1})\right)^{-1/2}(\sqrt{n_b}\|\pi-\widehat{\pi}\|_2 + \sqrt{n_b}\|\Pi-\widehat\Pi\|_{\operatorname{op}})\\
    &\leq \left(\tfrac{\lambda_{\min}(\widehat{\Sigma}_{Z_b})}{\lambda_{\min}(\widehat{\Sigma}_X)}\right)^{1/2}(\sqrt{n_b}\|\pi-\widehat{\pi}\|_2 + \sqrt{n_b}\|\Pi-\widehat\Pi\|_{\operatorname{op}}).
\end{align*}
where $\|\cdot\|_{\operatorname{op}}$ denotes the operator norm, and we used Weyl's inequality for the second inequality and that $\beta$ has norm one for the last inequality. Hence, using the bounds on the minimal eigenvalues of $\Sigma_X$ and $\Sigma_{Z_b}$ in Assumption~\ref{ass:regularity}, it holds that
\begin{equation}
\label{eq:sup_bound}
    \sup_{\beta\in\mathbb{R}^d:\|\beta\|_2=1}\|S(\beta)-\mu(\beta)\|_2=\mathcal{O}_{\mathcal{P}}(1)
\end{equation}
as $n_a,n_b$ tend to infinity, where $\mathcal{O}_{\mathcal{P}}(1)$ denotes a uniformly bounded random variable with respect to $\mathcal{P}$. Finally, for all $s\in[d]$ define $\overline{\cB}_s\coloneqq\{\beta\in\mathbb{R}^d\mid \|\beta\|_0=s \text{ and } \|\beta\|_2=1\}$. Then, using that $Q$ does not depend on the scale of $\beta$ and \eqref{eq:upper_bound_infQ} we get that
\begin{align}
    \mathbb{P}_P\left(\inf_{\beta:\|\beta\|_0=s}Q(\beta)\leq t\right)
    &=\mathbb{P}_P\left(\inf_{\beta\in\overline{\cB}_s}Q(\beta)\leq t\right)\nonumber\\
    &\leq 2\mathbb{P}_P\left(\sup_{\beta\in\mathbb{R}^d:\|\beta\|_2=1}\|S(\beta)-\mu(\beta)\|_2\geq \inf_{\beta\in\overline{\cB}_s}\|\mu(\beta)\|_2 - \sqrt{t}\right).\label{eq:inf_Q_bound}
\end{align}
Now for the first statement of Theorem~\ref{thm:consistency_subset}, fix $s \in\bbN$ such that $s < \norm{\beta^*}_0 = |\text{PA}(Y)|$.
It follows from Theorem~\ref{thm:identify_truebeta} that for all $\beta\in\bbR^d$ 
with $\norm{\beta}_0$ = s, $\pi - \Pi\beta \neq 0$. 
Therefore, there exists $\epsilon > 0$ such that for all $\beta \in\bbR^d$ with 
$\norm{\beta}_0 = s$, it holds that $\norm{\pi-\Pi\beta}_2 > \epsilon$. 
Therefore, by \eqref{eq:inf_Q_bound} it holds that
\begin{align*}
&\lim_{n_a,n_b\to\infty} \mathbb{P}_P\left(\varphi_s(\cD_{a,b}^{\text{joint}}) = 1 \right)\\ 
&= \lim_{n_a,n_b \to \infty} \mathbb{P}_P\left(\inf_{\beta: \norm{\beta}_0 = s} Q(\beta) > \kappa_m(1-\alpha)\right) \\ 
&\geq 1 - \lim_{n_a,n_b\to\infty} 2\mathbb{P}_P\left(\sup_{\beta\in\mathbb{R}^d:\|\beta\|_2=1}\|S(\beta)-\mu(\beta)\|_2\geq \inf_{\beta\in\overline{\cB}_s}\|\mu_{n_a,n_b}(\beta)\|_2 - \sqrt{\kappa_m(1-\alpha)}\right)\\
&= 1,
\end{align*}
where we used \eqref{eq:sup_bound} together with
\begin{align*}
    \lim_{n_a,n_b\to\infty}\inf_{\beta\in\overline{\cB}_s}\|\mu_{n_a,n_b}(\beta)\|_2
    &\geq \lim_{n_a,n_b\to\infty}\inf_{\beta\in\overline{\cB}_s}\sqrt{n_b}\|\tfrac{n_b}{n_a}\widehat{\Sigma}_\pi + \beta^{\top}\widehat{\Sigma}_X\beta\widehat{\Sigma}_{Z_b}^{-1}\|^{-1/2}_{\operatorname{op}}\epsilon\\
    &\geq \lim_{n_a,n_b\to\infty}\inf_{\beta\in\overline{\cB}_s}\left(\frac{1}{n_a}\lambda_{\max}(\widehat\Sigma_\pi)+\tfrac{1}{n_b}\tfrac{\lambda_{\max}(\widehat{\Sigma}_{X})}{\lambda_{\min}(\widehat{\Sigma}_{Z_b})}\right)^{-1/2}\epsilon\\
    &=\infty,
\end{align*}
where we again used the bounds on the minimal eigenvalues of $\Sigma_X$ and $\Sigma_{Z_b}$ in Assumption~\ref{ass:regularity}.
Since this holds for all $s\in[d]$ with $s < \norm{\beta^*}_0$, we further get 
\begin{align*}
\lim_{n_a,n_b\to\infty} \mathbb{P}_P\left(\norm{\beta^{\leq s_{\text{max}}}}_0 = \norm{\beta^*}_0\right) 
&= \lim_{n_a,n_b\to\infty} \mathbb{P}_P \left(\min_{s < \norm{\beta^*}_0} \varphi_s = 1 \text{ and } \varphi_{\norm{\beta^*}_0} = 0 \right) \\
&= \lim_{n_a,n_b\to\infty} \mathbb{P}_P \left(\varphi_{\norm{\beta^*}_0} = 0\right) \\
&\geq 1-\alpha.
\end{align*}
For the second statement of Theorem~\ref{thm:consistency_subset}, we can use the same argument. In this case, Theorem~\ref{thm:identify_truebeta} implies that for all $c>0$ there exists $\epsilon>0$ such that for all $\beta\in\mathbb{R}^d$ with either $\|\beta\|_0<\|\beta^*\|_0$ or $\|\beta-\beta^*\|>\epsilon$ and $\|\beta\|_0=\|\beta^*\|_0$ it holds that $\|\pi-\Pi\beta\|_2>\epsilon$. Therefore, $\mu_{n_a,n_b}(\beta)$ again diverges and the arguments above remain valid.
This completes the proof of Theorem~\ref{thm:consistency_subset}.
\end{proof}

\section{Additional results} \label{app:additional_results}

\begin{remark} \label{rmk:sig_pi_beta}
In the definition of the (empirical) Q statistic in Theorem~\ref{thm:q_stat}, we used 
\begin{align} \label{eq:sig_pi_beta_kronecker}
\widehat\Sigma_\Pi(\beta) \coloneqq \xi(\beta)\widehat\Sigma_\Pi\xi^\top(\beta) 
\end{align}
where $\xi(\beta) \coloneqq \beta^\top\otimes I_m$. 
It follows from the properties of Kronecker product that \eqref{eq:sig_pi_beta_kronecker} is equivalent to
\begin{align} \label{eq:sig_pi_beta_resid}
\widehat\Sigma_\Pi(\beta) \coloneqq (\beta^\top(\bX_b - \bZ_b\widehat\Pi)^\top (\bX_b - \bZ_b\widehat\Pi)\beta) (\bZ_b^\top\bZ_b)^{-1},
\end{align}
which aligns with its population quantity 
$\Sigma_\Pi(\beta) \coloneqq (\beta^{\top}\mathbb{E}[u_b^X(u_b^X)^{\top}]\beta) \mathbb{E}[Z_bZ_b^{\top}]^{-1}$ 
used in Lemma~\ref{lem:s_beta_asymptotic}, where $u_b^X$ is the population residual in \eqref{eq:reduced_scm}. 
The reason why \eqref{eq:sig_pi_beta_kronecker} is used instead of \eqref{eq:sig_pi_beta_resid} in the Q statistic 
is that \eqref{eq:sig_pi_beta_kronecker} only relies on the joint summary statistics, as the individual-level data is 
not available under the two-sample summary statistics setting.
\end{remark}

\begin{lemma} \label{lem:s_beta_asymptotic}
Assume Assumption~\ref{ass:regularity}. Let $\mathcal{D}_{a,b}^{\text{joint}}=\left\{\widehat\pi, \widehat\Sigma_\pi, \widehat\Pi, \widehat\Sigma_\Pi\right\}$ 
be the joint summary statistics based on two independent samples of sizes $n_a$ and $n_b$, respectively. For all  
$\beta\in\bbR^d$, define
\begin{equation*}
    S_{n_a, n_b}(\beta) \coloneqq \left(\tfrac{1}{n_a}\widehat\Sigma_\pi + \tfrac{1}{n_b}
\widehat\Sigma_\Pi(\beta) \right)^{-1/2}
(\widehat\pi - \widehat\Pi\beta)
\end{equation*}
and
\begin{equation*}
    \mu_{n_a,n_b}(\beta)\coloneqq\left(\tfrac{1}{n_a}\widehat\Sigma_\pi + \tfrac{1}{n_b}
\widehat\Sigma_\Pi(\beta) \right)^{-1/2}
(\pi - \Pi\beta).
\end{equation*}
Then, for all $\beta\in\mathbb{R}^d$ and all $r\in (0 ,\infty)$ it holds that
\begin{equation*}    \lim_{\substack{n_a,n_b\rightarrow\infty\\n_a/n_b\rightarrow r}}\sup_{P\in\mathcal{P}}\sup_{t\in\mathbb{R}^m} \left|\mathbb{P}_P\left(S_{n_a, n_b}(\beta)-\mu_{n_a,n_b}(\beta)\leq t\right) - \Phi_m(t)\right| = 0.
\end{equation*}
\end{lemma}

\begin{proof}
Fix an arbitrary $\beta\in\bbR^d$. Using by standard uniform convergence results for the OLS estimator \citepAPP[e.g.,][Lemma~S10]{lundborg2022projected} it holds that
\begin{equation*}
    \sqrt{n_a}\Sigma_{\pi}^{-1/2}(\widehat{\pi}-\pi) 
\end{equation*}
with $\Sigma_\pi \coloneqq \mathbb{E}[(u_a^Y)^2]\mathbb{E}[Z_a Z_a^{\top}]^{-1}$ (where $u_a^Y$ are the population residuals in \eqref{eq:reduced_scm} for sample $a$) converges uniformly w.r.t.\ $\mathcal{P}$ to a standard $m$-variate Gaussian distribution as $n_a$ tends to infinity. Similarly, when considering the regression of $\beta^{\top}X$ on $Z$, it holds that
\begin{equation*}
    \sqrt{n_b}\Sigma_{\Pi}(\beta)^{-1/2}(\widehat{\Pi}-\Pi)\beta
\end{equation*}
with $\Sigma_\Pi(\beta) \coloneqq (\beta^{\top}\mathbb{E}[u_b^X(u_b^X)^{\top}]\beta) \mathbb{E}[Z_bZ_b^{\top}]^{-1}$ 
(where $u_b^X$ are the residuals in \eqref{eq:reduced_scm} for sample $b$) converges uniformly w.r.t.\ $\mathcal{P}$ to a standard $m$-variate Gaussian distribution as $n_b$ tends to infinity. Combining these results and using that $n_a / n_b \to r$ and $\widehat\pi$ and $\widehat\Pi$ are estimated based on independent samples, we further have that
\begin{equation} \label{eq:diff_dist}
\sqrt{n_b}\left(\tfrac{1}{r}\Sigma_{\pi}+\Sigma_{\Pi}(\beta)\right)^{-1/2}\left((\widehat\pi - \widehat\Pi\beta) - (\pi - \Pi\beta)\right)
\end{equation}
converges uniformly w.r.t.\ $\mathcal{P}$ to a standard $m$-variate Gaussian distribution as $n_a$ and $n_b$ tend to infinity.

Next, we show for all $\epsilon>0$ that
\begin{equation}
\label{eq:var_estimate_conv}
\lim_{n_a\rightarrow\infty}\sup_{P\in\mathcal{P}}\mathbb{P}_P\left(\|\widehat\Sigma_\pi - \Sigma_\pi\|_{\operatorname{op}}>\epsilon\right)=0
\quad\text{and}\quad
\lim_{n_b\rightarrow\infty}\sup_{P\in\mathcal{P}}\mathbb{P}_P\left(\|\widehat\Sigma_\Pi(\beta) - \Sigma_\Pi(\beta)\|_{\operatorname{op}}>\epsilon\right)=0.
\end{equation}
As the proofs for both results are the same we only show it for $\widehat{\Sigma}_{\pi}$. First, we express the estimator as
\begin{align*}
    \widehat{\Sigma}_{\pi}=\frac{1}{n_a}\sum_{i=1}^{n_a}(Y_{ai}-\widehat{\pi}^{\top}Z_{ai})^2\left(\frac{1}{n_a}\sum_{i=1}^{n_a}Z_{ai}Z_{ai}^{\top}\right)^{-1}.
\end{align*}
We now consider the two product terms separately. Using the uniform law of large numbers \citepAPP[e.g.,][Lemma 9]{klyne2023average} on each component, it holds for all $\epsilon>0$ that
\begin{equation}
\label{eq:conv_covariance_matrix}
    \lim_{n_a\rightarrow\infty}\sup_{P\in\mathcal{P}}\mathbb{P}_P\left(\left\|\frac{1}{n_a}\sum_{i=1}^{n_a}Z_{ai}Z_{ai}^{\top}-\mathbb{E}[Z_a Z_a^{\top}]\right\|_{\operatorname{op}}>\epsilon\right)=0.
\end{equation}
Moreover, we can expand the residual variance part as follows
\begin{align*}
    \frac{1}{n_a}\sum_{i=1}^{n_a}(Y_{ai}-\widehat{\pi}^{\top}Z_{ai})^2
    &=\frac{1}{n_a}\sum_{i=1}^{n_a}(Y_{ai}-\pi^{\top}Z_{ai})^2+
    \frac{1}{\sqrt{n_a}}\left(\frac{2}{n_a}\sum_{i=1}^{n_a}(Y_{ai}-\pi^{\top}Z_{ai})\sqrt{n_a}(\widehat{\pi}-\pi)^{\top}Z_{ai}\right)\\
    &\qquad
    +\frac{1}{n_a}\left(\sqrt{n_a}(\widehat{\pi}-\pi)\left(\frac{1}{n_a}\sum_{i=1}^{n_a}Z_{ai}Z_{ai}^{\top}\right)\sqrt{n_a}(\widehat{\pi}-\pi)\right).
\end{align*}
Then, by the uniform asymptotic normality, the bounded moments of $Z$ and $Y$ and a further application of the law of large numbers \citepAPP[e.g.,][Lemma 9]{klyne2023average} it follows for all $\epsilon>0$ that 
\begin{equation}
\label{eq:conv_res_variance}
    \lim_{n_a\rightarrow\infty}\sup_{P\in\mathcal{P}}\mathbb{P}_P\left(\left|\frac{1}{n_a}\sum_{i=1}^{n_a}(Y_{ai}-\widehat{\pi}^{\top}Z_{ai})^2-\mathbb{E}[(u_a^Y)^2]\right| >\epsilon\right)=0.
\end{equation}
Finally, denote $W_n\coloneqq \frac{1}{n_a}\sum_{i=1}^{n_a}(Y_{ai}-\widehat{\pi}^{\top}Z_{ai})^2$, $W\coloneqq \mathbb{E}[(u^Y_a)^2]$, $V_n\coloneqq\frac{1}{n_a}\sum_{i=1}^{n_a}Z_{ai}Z_{ai}^{\top}$ and $V\coloneqq\mathbb{E}[Z_a Z_a^{\top}]$. Then, by combining \eqref{eq:conv_covariance_matrix} and \eqref{eq:conv_res_variance} it follows for all $\epsilon>0$ that
\begin{align*}
    &\sup_{P\in\mathcal{P}}\mathbb{P}_P(\|W_nV_n^{-1}-WV^{-1}\|_{\operatorname{op}}>\epsilon)\\
    &\quad\leq\sup_{P\in\mathcal{P}}\mathbb{P}_P(\|W_nV_n^{-1}-W_nV^{-1}\|_{\operatorname{op}}>\tfrac{\epsilon}{2})+\sup_{P\in\mathcal{P}}\mathbb{P}_P(\|W_nV^{-1}-WV^{-1}\|_{\operatorname{op}}>\tfrac{\epsilon}{2})\\
    &\quad\leq\sup_{P\in\mathcal{P}}\mathbb{P}_P(\|W_nV_n^{-1}-W_nV^{-1}\|_{\operatorname{op}}>\tfrac{\epsilon}{2})+\sup_{P\in\mathcal{P}}\mathbb{P}_P(\|W_nV^{-1}-WV^{-1}\|_{\operatorname{op}}>\tfrac{\epsilon}{2})\\
    &\quad\leq\sup_{P\in\mathcal{P}}\mathbb{P}_P(\|V_n^{-1}-V^{-1}\|_{\operatorname{op}}\|W_n\|_{\operatorname{op}}>\tfrac{\epsilon}{2})+\sup_{P\in\mathcal{P}}\mathbb{P}_P(\|W_n-W\|_{\operatorname{op}}\|V^{-1}\|_{\operatorname{op}}>\tfrac{\epsilon}{2})\\
    &\quad\leq\sup_{P\in\mathcal{P}}\mathbb{P}_P(\|V_n^{-1}-V^{-1}\|_{\operatorname{op}}\|W_n\|_{\operatorname{op}}>\tfrac{\epsilon}{2})+\sup_{P\in\mathcal{P}}\mathbb{P}_P(\|W_n-W\|_{\operatorname{op}}C>\tfrac{\epsilon}{2})\\
\end{align*}
By standard arguments and using the lower bound on the minimal eigenvalue of $V=\mathbb{E}[ZZ^{\top}]$ from Assumption~\ref{ass:regularity}, this proves \eqref{eq:var_estimate_conv} (left).

Combining the two convergence results in \eqref{eq:var_estimate_conv} and using that $n_a/n_b\to r$ shows that for all $\epsilon>0$ it holds that
\begin{equation} \label{eq:conv_sigmas}
    \lim_{\substack{n_a,n_b\rightarrow\infty\\n_a/n_b\rightarrow r}}\sup_{P\in\mathcal{P}}\mathbb{P}_P\left(\|(\tfrac{n_b}{n_a}\widehat\Sigma_\pi + \widehat{\Sigma}_{\Pi}(\beta)) - (\tfrac{1}{r}\Sigma_\pi+\Sigma_{\Pi}(\beta))\|_{\operatorname{op}}>\epsilon\right)=0.
\end{equation}
Furthermore, we can apply \citetAPP[][eq.~(7.2.13)]{johnson1985matrix} to get that
\begin{align*}
&\|(\tfrac{n_b}{n_a}\widehat\Sigma_\pi + \widehat\Sigma_\Pi(\beta))^{1/2} -  (\tfrac{1}{r}\Sigma_\pi + \Sigma_\Pi(\beta))^{1/2}\|_{\operatorname{op}} \\
&\leq \|(\tfrac{1}{r}\Sigma_\pi + \Sigma_\Pi(\beta))^{-1/2}\|_{\operatorname{op}} \|(\tfrac{n_b}{n_a}\widehat\Sigma_\pi + \widehat\Sigma_\Pi(\beta)) -  (\tfrac{1}{r}\Sigma_\pi + \Sigma_\Pi(\beta))\|_{\operatorname{op}},
\end{align*}
which together with \eqref{eq:conv_sigmas} and  since Assumption~\ref{ass:regularity} implies that $\inf_{P\in\mathcal{P}}\lambda_{\min}(\frac{1}{r}\Sigma_\pi + \Sigma_\Pi(\beta))>0$, implies for all $\epsilon > 0$ that
\begin{align*}
\lim_{\substack{n_a,n_b\rightarrow\infty\\n_a/n_b\rightarrow r}}\sup_{P\in\mathcal{P}}\mathbb{P}_P(\|(\tfrac{n_b}{n_a}\widehat\Sigma_\pi + \widehat\Sigma_\Pi(\beta))^{1/2} -  (\tfrac{1}{r}\Sigma_\pi + \Sigma_\Pi(\beta))^{1/2}\|_{\operatorname{op}} > \epsilon) = 0. 
\end{align*}
Together with \eqref{eq:diff_dist} this implies by \citetAPP[Lemma~10 (b)]{klyne2023average} that
\begin{align*}
    \lim_{\substack{n_a,n_b\rightarrow\infty\\n_a/n_b\rightarrow r}}\sup_{P\in\mathcal{P}}\sup_{t\in\mathbb{R}^m}\left|\mathbb{P}_P\left(S_{n_a,n_b}(\beta)-\mu_{n_a,n_b}(\beta)\leq t\right)-\Phi_m(t)\right|=0,
\end{align*}
where we in particular used that
\begin{align*}
&S_{n_a,n_b}(\beta)-\mu_{n_a,n_b}(\beta)\\
&\quad=(\tfrac{1}{n_a}\widehat{\Sigma}_{\pi}+\tfrac{1}{n_b}\widehat{\Sigma}_{\Pi}(\beta))^{-1/2}((\widehat{\pi}-\widehat{\Pi}\beta) - (\pi-\Pi\beta)) \\
&\quad=\left((\tfrac{1}{r}\Sigma_{\pi}+\Sigma_{\Pi}(\beta))^{-1/2}(\tfrac{n_b}{n_a}\widehat{\Sigma}_{\pi}+\widehat{\Sigma}_{\Pi}(\beta))^{1/2}\right)^{-1}
    \sqrt{n_b}(\tfrac{1}{r}\Sigma_{\pi}+\Sigma_{\Pi}(\beta))^{-1/2}((\widehat{\pi}-\widehat{\Pi}\beta) - (\pi-\Pi\beta))\\
&\quad= \left(I + 
(\tfrac{1}{r}\Sigma_\pi + \Sigma_\Pi(\beta))^{-1/2}\{(\tfrac{n_b}{n_a}\widehat\Sigma_\pi + \widehat\Sigma_\Pi(\beta))^{1/2} -  (\tfrac{1}{r}\Sigma_\pi + \Sigma_\Pi(\beta))^{1/2}\}\right)^{-1}\\
&\quad\qquad\qquad\qquad\qquad\qquad\qquad\qquad\qquad\cdot\sqrt{n_b}(\tfrac{1}{r}\Sigma_{\pi}+\Sigma_{\Pi}(\beta))^{-1/2}((\widehat{\pi}-\widehat{\Pi}\beta) - (\pi-\Pi\beta)).
\end{align*}
This completes the proof of Lemma~\ref{lem:s_beta_asymptotic}.
\end{proof}

\section{Experiment details and additional simulation results} \label{app:experiment}

\subsection{Details of the simulated experiments in Section~\ref{sec:simulation}} \label{app:experiment_detail}

\textbf{DGP1:} The individual-level data are generated from an SCM \eqref{eq:individual_scm} with 
the following parameters
\begin{align*}
A \coloneqq \begin{pmatrix}
0 & 0 & 0 \\
0 & 0 & 0 \\
1 & 0 & 0 \\
0 & 1 & 0 \\
0 & 0 & 1
\end{pmatrix} 
\quad
\text{and}
\quad
B \coloneqq \begin{pmatrix}
0 & 0 & 0 & 0 & 0 \\
0 & 0 & 0 & 0 & 0 \\
1 & 0 & 0 & 0 & 0 \\
1 & 1 & 0 & 0 & 0 \\
0 & 1 & 0 & 0 & 0,
\end{pmatrix}
\end{align*}
$Z \iid \cN_m(0, I_m)$, $H \iid \cN_d(0, I_d)$ and $\nu^X, \nu^Y\iid\cN(0,1)$ 
with $g(H, \nu^X) \coloneqq H + \nu^X$ and 
$h(H, \nu^Y) \coloneqq H^\top\bm{1}_d + \nu^Y$. 
The true causal effect $\beta^* = (1, 2, 0, 0, 0)$.

\noindent\textbf{DGP2:} Let 
\begin{align*}
A \coloneqq \begin{pmatrix}
0 & 0 & 0 \\
0 & 0 & 0 \\
1 & 0 & 0 \\
0 & 1 & 0 \\
0 & 0 & 1
\end{pmatrix},
\hspace{5pt}
B \coloneqq \begin{pmatrix}
0 & 0 & 0 & 0 & 0 \\
0 & 0 & 0 & 0 & 0 \\
1 & 0 & 0 & 0 & 0 \\
1 & 1 & 0 & 0 & 0 \\
0 & 2 & 0 & 0 & 0,
\end{pmatrix},
\hspace{5pt}
\var(Z) \coloneqq \begin{pmatrix}
1 & 0.05 & -0.1 & 0.075 & 0.025 \\ 
0.05 & 1 & 0 & 0 & 0 \\
-0.1 & 0 & 1 & 0 & 0 \\ 
0.075 & 0 & 0 & 1 & 0 \\
0.025 & 0 & 0 & 0 & 1
\end{pmatrix},
\end{align*}
$\var(\nu^X) \coloneqq WW^\top + I_d$ 
where $W \in\bbR^{d\times d}$ with $W_i^j \iid \text{Unif}(-0.3, 0.5)$ for all $i,j\in[d]$, 
$\var(\nu^Y) \coloneqq 1$, and $\cov(\nu^X, \nu^Y) \in\bbR^{100}$ such that 
$\cov(\nu^X, \nu^Y)^j$ is uniformly sampled from the set $\{0.2, 0.4, 0.6, 0.8\}$ 
for all $j\in\{1,\ldots,100\}$.

Then using $\beta^* \in\bbR^{100}$ with $(\beta^*)^1 \coloneqq 1$, $(\beta^*)^2 \coloneqq 2$, and 
$(\beta^*)^j \coloneqq 0$ for all $j\in\{3,\ldots,100\}$, we 
define $\Pi \coloneqq A^\top (I_d - B)^{-1}$ and $\pi \coloneqq \Pi\beta^*$. Moreover, 
based on the linear SCM and with $V \coloneqq (I_d - B)^{-1} \beta^*$ we have 
\begin{align*}
\Sigma_\pi &= \left(V^\top \var(\nu^X) V + \var(\nu^Y) + 2 V^\top\cov(\nu^X, \nu^Y)\right) \var(Z)^{-1} \text{ and } \\
\Sigma_\Pi &= \var(\nu^X)^{\top} (I_d - B)^{-1} \otimes \var(Z)^{-1}.
\end{align*} 
We then generated $\widehat\pi$, $\widehat\Pi$ from the following  
multivariate Gaussian distributions for a specific sample size $n$:
\begin{align*}
\widehat\pi_n \sim \cN(\pi, \frac{1}{n}\Sigma_\pi) \qquad \text{and} \qquad
\widehat\Pi_n \sim \cN(\Pi, \frac{1}{n}\Sigma_\Pi). 
\end{align*}

\subsection{Additional simulated experiments}

We provide additional simulation results of a setting that the 
exclusion restriction criteria of IV is violated. The DGP is 
described below and the corresponding DAG is given in Figure~\ref{fig:dag_dgp3}.

\paragraph{DGP3:} $m = 5$ and $d = 5$ and $\norm{\beta^*}_0 = 2$. For increasing 
$n \coloneqq n_1 = n_2$, we generate iid $\{(Y_i, Z_i)\}_{i=1}^{n_1}$ and 
$\{(X_i, Z_i)\}_{i=1}^{n_2}$ according to the following SCM 
\begin{align} \label{eq:scm_dgp3}
\begin{split}
X_i &\coloneqq AZ_i + BX_i + H_i + \nu_i^X \\
Y_i &\coloneqq X_i^\top \beta^* + Z_i^\top\gamma + H_i^\top\bm{1}_5 + \nu_i^Y,
\end{split}
\end{align}
with the following parameters:
\begin{align*}
A &= I_5, \quad
B \coloneqq \begin{pmatrix}
0 & 0 & 0 & 0 & 0 \\
0 & 0 & 0 & 0 & 0 \\
1 & 0 & 0 & 0 & 0 \\
1 & 1 & 0 & 0 & 0 \\
0 & 1 & 0 & 0 & 0
\end{pmatrix}, \quad
\gamma = (0.1, 0.1),  \quad
\beta^* = (1, 2, 0, 0, 0), 
\end{align*}
$H_i \iid\cN(0, I_5)$, and $\nu_i^X, \nu_i^Y \iid\cN(0,1)$. 
Then we compute the summary statistics using seemingly unrelated regression. 
The results are shown in  Figure~\ref{fig:dgp3_eval}. 
The $\gamma$ parameter in \eqref{eq:scm_dgp3} represents the violation of the exclusion restriction 
criteria. We see that as sample size goes larger, the bias and rmse 
continue to decrease. Although the Jaccard similarity and percentage of 
correct size start to decline, the average true positive rate (tpr) still 
stays around $100\%$. In this example, due to the invalid instruments,
the estimated causal parents tend to be a superset of the true causal
parent, but the estimated effects of the non-parent covaraites are 
relatively small. 

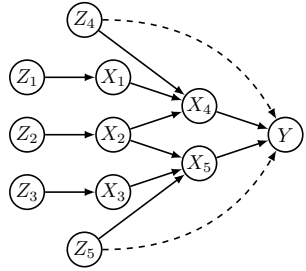
\begin{figure}
\centering
  \resizebox*{0.27\textwidth}{!}{
    \begin{tikzpicture}
    % Style definitions
    \tikzstyle{var}=[circle, draw, thick, minimum size=6mm, font=\small, inner sep=0]
    \tikzstyle{arrow}=[-latex, thick]
    \tikzstyle{doublearrow}=[latex-latex, thick]
    \tikzstyle{dashedarrow}=[-latex, thick, dashed]

    % Nodes for I_1 to I_J (G variables in sketch)
    \node[var] (I4) at (1, 4) {$Z_4$};
    \node[var] (I5) at (1, 0) {$Z_5$};
    \node[var] (I1) at (0, 3) {$Z_1$};
    \node[var] (I2) at (0, 2) {$Z_2$};
    \node[var] (I3) at (0, 1) {$Z_3$};

    % Nodes for X_1 to X_R
    \node[var] (X4) at (3, 2.5) {$X_4$};
    \node[var] (X5) at (3, 1.5) {$X_5$};
    \node[var] (X1) at (1.5, 3) {$X_1$};
    \node[var] (X2) at (1.5, 2) {$X_2$};
    \node[var] (X3) at (1.5, 1) {$X_3$};

    % Node for H
    % \node[var] (H) at (3.2, 3) {$H$};

    % Node for Y
    \node[var] (Y) at (4.5, 2) {$Y$};

    % % Double-arrow arcs between I variables (G variables in sketch)
    % \draw[doublearrow, bend right=45] (I1) to (I3);
    % \draw[doublearrow, bend left=45] (I2) to (I5);
    % \draw[doublearrow, bend right=45] (I3) to (I4);
    % \draw[doublearrow, bend right=45] (I4) to (I5);
    % \draw[doublearrow, bend right=45] (I3) to (I5);

    % % Arrows from some I variables to some X variables
    \draw[arrow] (I1) -- (X1);
    \draw[arrow] (I2) -- (X2);
    \draw[arrow] (I3) -- (X3);
    \draw[arrow] (I4) -- (X4);
    \draw[arrow] (I5) -- (X5);
    \draw[dashedarrow, bend left] (I4) to (Y);
    \draw[dashedarrow, bend right] (I5) to (Y);

    % % H points to both X1 and Y
    % \draw[arrow] (H) -- (X1);
    % \draw[arrow] (H) -- (X2);
    % \draw[arrow] (H) -- (XR);
    % \draw[arrow] (H) -- (Y);

    % Some X points to X
    \draw[arrow] (X1) -- (X4);
    \draw[arrow] (X2) -- (X4);
    \draw[arrow] (X2) -- (X5);
    \draw[arrow] (X3) -- (X5);

    % Some X points to Y
    \draw[arrow] (X4) -- (Y);
    \draw[arrow] (X5) -- (Y);

\end{tikzpicture}
  } 
\caption{DAG for DGP3 which contains two invalid instruments violating 
the exclusion restriction criteria (dashed arrows).}
\label{fig:dag_dgp3}
\end{figure}

\begin{figure}[!htb]
\centering
\includegraphics[width=0.9\linewidth]{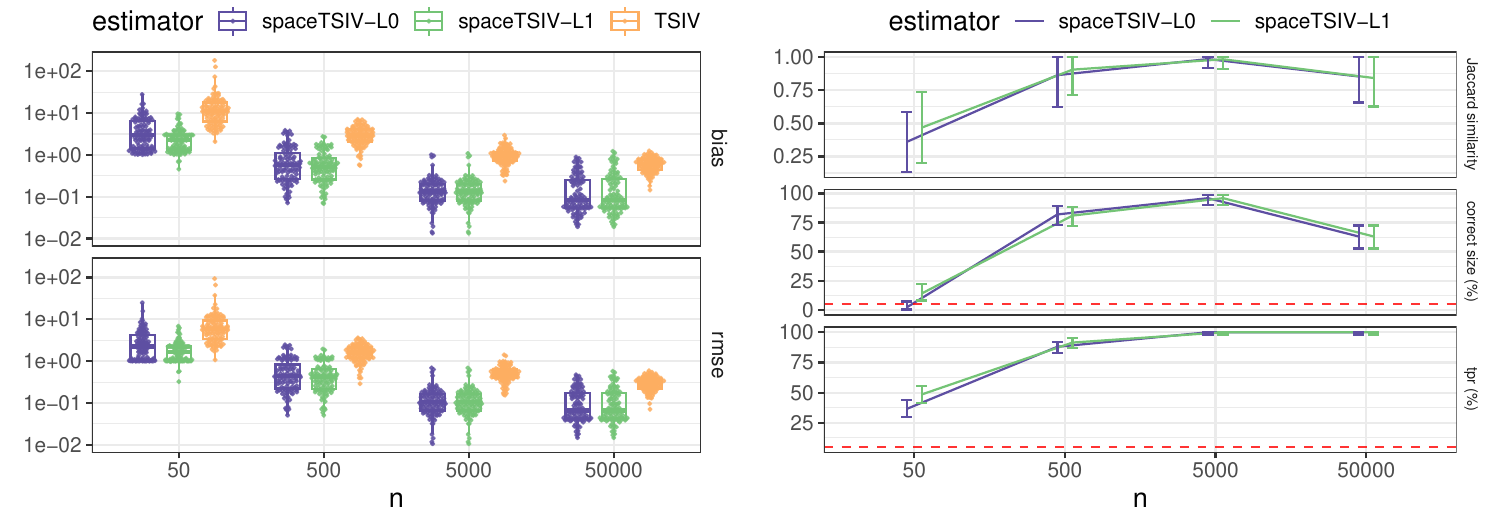}
\caption{Results using data generated by DGP3 based on $100$ repetitions. 
Left: Bias and rmse of the estimators. The y-axis is on log scale for clarity. 
Right: Average Jaccard similarity between the selected covariates and 
the true causal covariates (error bars indicate confidence intervals 
constructed by mean plus/minus one standard error), 
percentage of estimates that have the correct support size, 
and tpr 
(error bars indicate $95\%$ binomial confidence intervals).
This DGP contains $2$ invalid instruments among the $5$ instruments.}
\label{fig:dgp3_eval}
\end{figure}

\newpage

\bibliographystyleAPP{abbrvnat}
\bibliographyAPP{bibliography}

\end{document}